\documentclass[10pt]{article}
\usepackage[margin=1in]{geometry}
\usepackage[english]{babel}
\usepackage{setspace}
\usepackage{mathrsfs,hyperref, algorithm, algorithmic}
\usepackage[]{amsmath}
\usepackage{amsthm}
\usepackage{fix-cm}
\usepackage[]{amssymb}
\usepackage[]{latexsym}
\usepackage[latin1]{inputenc}
\usepackage[right]{eurosym}
\usepackage[T1]{fontenc}
\usepackage[]{graphicx}
\usepackage[]{epsfig}
\usepackage{fancyhdr}
\usepackage{multirow}
\usepackage{natbib}
\usepackage{subcaption}
\usepackage{tikz}
\usetikzlibrary{decorations.pathreplacing}
\usetikzlibrary{arrows,automata,positioning}
\makeatletter

\newcommand{\bbr}{\mathbb{R}}

\newcommand{\bbn}{\mathbb{N}}

\newcommand{\fn}{\footnote}
\newcommand{\ci}{\cite}

\newcommand{\Alpha}{A}

\newcounter{modcount}
\newcommand{\modulo}[2]{%
\setcounter{modcount}{#1}\relax
\ifnum\value{modcount}<#2\relax
\else\relax
\addtocounter{modcount}{-#2}\relax
\modulo{\value{modcount}}{#2}\relax
\fi}

\newcommand{\tablepictures}[4][c]{\begin{tabular}[#1]{@{}c@{}}#2\vspace{0.5cm}\\(\alph{#4}) #3\end{tabular}}
\newcounter{gridsearch}
\newcommand{\tabpic}[2]{
    \stepcounter{gridsearch}
    \modulo{\thegridsearch}{2}
    \ifnum\value{modcount}=0
        \tablepictures[t]{#1}{#2}{gridsearch}\\[2.0cm]
    \else
        \tablepictures[t]{#1}{#2}{gridsearch}&~&
    \fi
}

\makeatother
\newtheorem{lemma}{Lemma}[section]
\newtheorem{proposition}[lemma]{Proposition}

\newtheorem{corollary}[lemma]{Corollary}

\newtheorem{example1}[lemma]{Example}
\newtheorem{rem1}[lemma]{Remark}
\newtheorem{assumption}[lemma]{Assumption}
\newtheorem{alg1}[lemma]{Algorithm}
\newtheorem{me1}[lemma]{Mechanism}

\newenvironment{remark}{\begin{rem1}\rm}{\end{rem1}}
\newenvironment{example}{\begin{example1}\rm}{\end{example1}}

\newenvironment{alg}{\begin{alg1}\rm}{\end{alg1}}

\usepackage{color}

\newcommand{\T}{\mathsf{T}}
\newcommand{\diag}{\operatorname{diag}}

\DeclareMathOperator*{\argmax}{arg\,max}

\begin{document}

\title{Obligations with Physical Delivery in a Multi-Layered Financial Network\footnote{The author would like to thank Eric Schaanning for his help in calibrating the Eisenberg-Noe network model to the European banking dataset in Example~\ref{Ex:Greece}.}}
\author{Zachary Feinstein\footnote{Zachary Feinstein, ESE, Washington University, St. Louis, MO 63130, USA, {\tt zfeinstein@ese.wustl.edu}.}\\[0.7ex] \textit{Washington University in St. Louis}}
\date{\today}
\maketitle

\begin{abstract}
%
This paper provides a general framework for modeling financial contagion in a system with obligations in multiple illiquid assets (e.g., currencies).  In so doing, we develop a multi-layered financial network that extends the single network of \cite{EN01}.  In particular, we develop a financial contagion model with fire sales that allows institutions to both buy and sell assets to cover their liabilities in the different assets and act as utility maximizers.  

We prove that, under standard assumptions and without market impacts, equilibrium portfolio holdings exist and are unique.  However, with market impacts, we prove that equilibrium portfolio holdings and market prices exist which clear the multi-layered financial system.  In general, though, these clearing solutions are not unique.  We extend this result by considering the t\^atonnement process to find the unique attained equilibrium.  The attained equilibrium need not be continuous with respect to the initial shock; these points of discontinuity match those stresses in which a financial crisis becomes a systemic crisis.  We further provide mathematical formulations for payment rules and utility functions satisfying the necessary conditions for these existence and uniqueness results.

We demonstrate the value of our model through illustrative numerical case studies.  In particular, we study a counterfactual scenario on the event that Greece re-instituted the drachma on a dataset from the European Banking Authority.
\end{abstract}\vspace{0.2cm}

\section{Introduction}\label{Sec:intro}

As defined in \cite{feinstein2015illiquid}, ``financial contagion occurs when the distress of one bank jeopardizes the health of other financial firms.''  Many recent works on the topic have focused on modeling aspects of the 2007-2009 financial crisis, as that event proved that systemic crises can have terrible costs.  However, such contagious events have occurred at other times in the recent past, e.g., the 1997 Asian financial crisis.  In that crisis, among others, currency fluctuations between the US dollar and, e.g., Thai baht and Indonesian rupiah caused the debt-to-income ratios of firms to jump.  This caused a positive feedback loop in the currency fluctuations, thus intensifying the contagion.  
In fact, \cite{DLM11} and references therein showed that foreign currency obligations for banks statistically increase the chance of a banking crisis in a nation.
However, in contrast to the financial contagion models of \cite{EN01,feinstein2015illiquid}, many historical financial crisis involved obligations and incomes in multiple currencies (that must be fulfilled in the quoted currency) and illiquidity in the currency markets (see, e.g., \cite{ABK03}).  That is, in a general sense, many historical crises exist as the outcome of a multi-layered financial network of obligations between financial institutions in multiple illiquid assets insofar as they exhibit three key components: (i) distinct networks of interbank obligations in each currency with (ii) intra-layer connections via payments made in the individual currencies and (iii) inter-layer interactions through asset transfers (and price impacts) between the different currencies and layers of the network.  As such, this current paper will focus on an extension of \cite{EN01} to allow for a multi-layered network of obligations, notably allowing for firms to transfer wealth between multiple (illiquid) assets or currencies causing price impacts to the exchange rates.

\cite{EN01} propose a network model for the spread of defaults in the financial system.  In that proposed model, banks hold liquid assets which are used to pay off liabilities; unpaid liabilities may infect additional firms and cause them to default on some of their liabilities as well.  That paper proves conditions for existence and uniqueness of the clearing payments and provides a method for computing the equilibrium clearing payments.  
This model has been extended to account for time dynamics in, e.g., \cite{KV16,CC15,BBF18}.
Additionally, the basic clearing payment model of \cite{EN01} has been relaxed to consider bankruptcy costs (e.g., \cite{E07,RV13,EGJ14,GY14,AW_15,CCY16}) and cross-holdings (e.g., \cite{E07,EGJ14,AW_15}).
Illiquid assets and fire sale dynamics have been included in the setting of such network models in, e.g., \cite{CFS05,NYYA07,GK10,AFM13,CLY14,AW_15,AFM16} for a single (representative) asset and \cite{feinstein2015illiquid,feinstein2016leverage} for multiple assets. 
Empirical studies of the aforementioned financial contagion models have been conducted in, e.g., \cite{ELS06,U11,CMS10,GY14}.  One of the key contributions of these works is the conclusion that the local connections, via contractual liabilities, do not capture most financial contagion.  This motivates our current study for considering the role of illiquid assets and currencies on financial contagion.
Measures of systemic risk have been studied in, e.g., \cite{chen2013axiomatic,kromer2013systemic,feinstein2014measures,BFFM2015,ACDP2015}.

The main advance that we wish to study in these models is a more complete picture of how illiquid assets impact systemic risk; in particular, we are concerned with the implications of physical obligations in multiple currencies.  Within this scope, prior work has focused on fire sales in which the various financial firms will liquidate their assets in case of a cash shortfall thus driving down the asset value.  
Further, we will demonstrate that such modeling inherently produces systemic crises due to the market switching the attained equilibria and thus having jumps in the response of banks and the market.
However, to the best of our knowledge, none of the prior works in the \cite{EN01} setting permit a multi-layered or interconnected financial network of obligations with liabilities in multiple assets and payments in physical assets rather than in some num\'eraire.  We refer to \cite{BK16OFRbrief} for a discussion of why multi-layered systems are important and can affect financial contagion; this is of particular importance due to the role that currency movements have on systemic risk (\cite{DLM11}).  In such a model, firms no longer only have obligations in the num\'{e}raire asset (cash) only, but in, e.g., multiple currencies or securities.  We refer to~\cite{MK13,BCD16,poledna15multilayer} for other approaches to modeling interconnected financial networks.
We will tackle this problem, and extend works further, by allowing for \emph{solvent} firms to invest via a utility maximization problem, thereby permitting such firms to purchase assets at the fire sale price.  

The systemic risk studied in this paper, insofar as it relates to currency crises, should be viewed as studying extreme events such as the abandonment of a currency peg.  If symmetry exists to balance the buying and selling of a currency (due to the notion of ``buy low and sell high'' as depicted by the wealth maximizing utility in Example~\ref{Ex:2asset_compare}), the exchange rates will generally be very stable as illustrated by low volatility in the foreign exchange markets.  However, when an asymmetry exists between initial holdings (in a local currency) and obligations (in a major currency, e.g., US dollars), and with atypical monetary policy, the multilayered network can cause large fluctuations in exchange rates.  It is this latter scenario we concern ourselves with in this work.

Notably, we will prove only the existence of clearing solutions in markets with price impacts.  This is comparable to, e.g.,~\cite{CFS05}, in which the clearing payments are not necessarily unique due to the introduction of fire sales.  In order to determine the attained clearing solution, we introduce the use of the t\^atonnement process which has previously been used to study financial contagion in~\cite{BFK14}.  Of particular interest, due to the nonuniqueness of the clearing solutions, the attained equilibria need not be continuous with respect to the initial shock.  Specifically, this jump implies that a small perturbation in shock can greatly influence the attained clearing solution.  This has far reaching consequences for stress testing as discrete stresses may miss this point of discontinuity and thus underestimate the true risk of financial contagion.

The organization of this paper is as follows.  First, in Section~\ref{Sec:setting}, we will introduce the mathematical and financial setting.  In Section~\ref{Sec:model}, we first develop the general modeling framework.  We consider markets without price impacts (Section~\ref{Sec:monotonic}) followed by markets with price impacts (Section~\ref{Sec:mAsset}).  We find conditions so that there exists a clearing solution (Section~\ref{Sec:mAsset}) and consider the t\^atonnoment process to find the attained market equilibrium (Section~\ref{Sec:attained}). These are the main results of this work.  Section~\ref{Sec:examples} is used to provide specific mathematical examples with meaningful financial interpretations that fit the results of Sections~\ref{Sec:model}.  
Numerical case studies are provided in Section~\ref{Sec:casestudy}.  In particular, beyond demonstrating the impact of differing choices of payment rules and utility parameters on the equilibrium response in toy models, we provide a numerical case study to consider the impacts on contagion of having a single currency split into two.  This is meaningful as it has been threatened in recent years for the Greek economy; in studying the so-called Grexit event, we calibrate the financial system to 2011 stress testing data from the European Banking Authority.
The proofs of the theoretical results are provided in the appendix.

\section{The stylized balance sheet}\label{Sec:setting}
\begin{figure}[h!]
\centering
\begin{subfigure}[b]{.47\textwidth}
\begin{tikzpicture}[xscale=1.15]
\draw[draw=none] (0,8.5) rectangle (6,9) node[pos=.5]{\bf Marked-to-Market Book Value};
\draw[draw=none] (0,8) rectangle (3,8.5) node[pos=.5]{\bf Assets};
\draw[draw=none] (3,8) rectangle (6,8.5) node[pos=.5]{\bf Liabilities};

\filldraw[fill=blue!20!white,draw=black] (0,6) rectangle (3,8) node[pos=.5,style={align=center}]{Endowment (Asset 1) \\ $q_1 x_i^1$};
\filldraw[fill=blue!10!white,draw=black] (0,3.75) rectangle (3,6) node[pos=.5,style={align=center}]{Interbank (Asset 1) \\ $q_1 \sum_{j = 1}^n L_{ji}^1$};
\filldraw[fill=violet!20!white,draw=black] (0,2.5) rectangle (3,3.75) node[pos=.5,style={align=center}]{Endowment (Asset 2) \\ $q_2 x_i^2$};
\filldraw[fill=violet!10!white,draw=black] (0,0) rectangle (3,2.5) node[pos=.5,style={align=center}]{Interbank (Asset 2) \\ $q_2 \sum_{j = 1}^n L_{ji}^2$};

\filldraw[fill=red!20!white,draw=black] (3,4) rectangle (6,8) node[pos=.5,style={align=center}]{Total (Asset 1) \\ $q_1 \bar p_i^1$};
\filldraw[fill=orange!20!white,draw=black] (3,2.5) rectangle (6,4) node[pos=.5,style={align=center}]{Total (Asset 2) \\ $q_2 \bar p_i^2$};
\filldraw[fill=yellow!20!white,draw=black] (3,0) rectangle (6,2.5) node[pos=.5,style={align=center}]{Book Capital};
\end{tikzpicture}
\caption{Marked-to-market book for firm $i$, i.e.\ assuming all liabilities are paid in full.}
\label{Fig:bs-book}
\end{subfigure}
~
\begin{subfigure}[b]{.47\textwidth}
\begin{tikzpicture}[xscale=1.15]
\draw[draw=none] (0,8.5) rectangle (6,9) node[pos=.5]{\bf Marked-to-Market Realized Balance Sheet};
\draw[draw=none] (0,8) rectangle (3,8.5) node[pos=.5]{\bf Assets};
\draw[draw=none] (3,8) rectangle (6,8.5) node[pos=.5]{\bf Liabilities};

\filldraw[fill=blue!20!white,draw=black] (0,6) rectangle (3,8) node[pos=.5,style={align=center}]{Endowment (Asset 1) \\ $q_1 x_i^1$};
\filldraw[fill=blue!10!white,draw=black] (0,5) rectangle (3,6) node[pos=.5,style={align=center}]{Interbank (Asset 1) \\ $q_1 \sum_{j = 1}^n a_{ji}^1 [\bar p_j^1 \wedge y_j^1]$};
\filldraw[fill=violet!20!white,draw=black] (0,3.75) rectangle (3,5) node[pos=.5,style={align=center}]{Endowment (Asset 2) \\ $q_2 x_i^2$};
\filldraw[fill=violet!10!white,draw=black] (0,2) rectangle (3,3.75) node[pos=.5,style={align=center}]{Interbank (Asset 2) \\ $q_2 \sum_{j = 1}^n a_{ji}^2 [\bar p_j^2 \wedge y_j^2]$};
\filldraw[fill=blue!10!white,draw=none] (0,0.75) rectangle (3,2);
\filldraw[fill=violet!10!white,draw=none] (0,0) rectangle (3,0.75);
\draw (0,2) -- (3,2);
\draw[dotted] (0,0.75) -- (3,0.75);

\filldraw[fill=red!20!white,draw=black] (3,4) rectangle (6,8) node[pos=.5,style={align=center}]{Total (Asset 1) \\ $q_1 \bar p_i^1$};
\filldraw[fill=orange!20!white,draw=black] (3,2.5) rectangle (6,4) node[pos=.5,style={align=center}]{Total (Asset 2) \\ $q_2 \bar p_i^2$};
\filldraw[fill=yellow!20!white,draw=black] (3,2) rectangle (6,2.5) node[pos=.5,style={align=center}]{Realized Capital};
\filldraw[fill=yellow!20!white,draw=black] (3,0) rectangle (6,2);
\draw (0,0) rectangle (6,8);
\draw (3,0) -- (3,8);

\begin{scope}
    \clip (0,0) rectangle (6,2);
    \foreach \x in {-6,-5.5,...,6}
    {
        \draw[line width=.5mm] (\x,0) -- (6+\x,8);
    }
\end{scope}
\end{tikzpicture}
\caption{Marked-to-market balance sheet for firm $i$ before transferring assets to cover liabilities.}
\label{Fig:bs-balance}
\end{subfigure}

\vspace{1cm}

\begin{subfigure}[b]{.47\textwidth}
\hspace{-2cm}
\begin{tikzpicture}[xscale=1.15]
\draw[draw=none] (0,8.5) rectangle (6,9) node[pos=.5]{\bf Marked-to-Market Realized Holdings and Balance Sheet};
\draw[draw=none] (0,8) rectangle (3,8.5) node[pos=.5]{\bf Assets};
\draw[draw=none] (3,8) rectangle (6,8.5) node[pos=.5]{\bf Liabilities};

\filldraw[fill=blue!15!white,draw=black] (0,4) rectangle (3,8) node[pos=.5,style={align=center}]{Holdings (Asset 1) \\ $q_1 y_i^1$};
\draw[dashed] (0,5) -- (3,5);
\filldraw[fill=violet!15!white,draw=black] (0,2) rectangle (3,4) node[pos=.5,style={align=center}]{Holdings (Asset 2) \\ $q_2 y_i^2$};
\draw[decorate,decoration={brace,amplitude=7pt},xshift=-4pt,yshift=0pt] (0,4) -- (0,5) node [black,midway,xshift=-1.2cm,style={align=center}]{Transferred \\ Assets \\ $2 \; \Rightarrow \; 1$};
\filldraw[fill=blue!10!white,draw=none] (0,0.75) rectangle (3,2);
\filldraw[fill=violet!10!white,draw=none] (0,0) rectangle (3,0.75);
\draw (0,2) -- (3,2);
\draw[dotted] (0,0.75) -- (3,0.75);

\filldraw[fill=red!20!white,draw=black] (3,4) rectangle (6,8) node[pos=.5,style={align=center}]{Total (Asset 1) \\ $q_1 \bar p_i^1$};
\filldraw[fill=orange!20!white,draw=black] (3,2.5) rectangle (6,4) node[pos=.5,style={align=center}]{Total (Asset 2) \\ $q_2 \bar p_i^2$};
\filldraw[fill=yellow!20!white,draw=black] (3,2) rectangle (6,2.5) node[pos=.5,style={align=center}]{Realized Capital};
\filldraw[fill=yellow!20!white,draw=black] (3,0) rectangle (6,2);
\draw (0,0) rectangle (6,8);
\draw (3,0) -- (3,8);

\begin{scope}
    \clip (0,0) rectangle (6,2);
    \foreach \x in {-6,-5.5,...,6}
    {
        \draw[line width=.5mm] (\x,0) -- (6+\x,8);
    }
\end{scope}
\end{tikzpicture}
\caption{Marked-to-market balance sheet for firm $i$ after transferring assets to cover liabilities.}
\label{Fig:bs-holdings}
\end{subfigure}
\caption{Stylized balance sheet for firm $i$ under price vector $q$.}
\label{Fig:balance_sheet}
\end{figure} 

Consider a financial system with $n$ financial institutions (e.g., banks, hedge funds, or pension plans) and a financial market with $m$ illiquid assets (possibly currencies).  We denote by $y \in \bbr^{n \times m}_+$ the realized portfolio holdings of the institutions and by $q \in \bbr^m_{++}$ the prices of the assets in some num\'{e}raire.  We assume throughout that the price of each asset is bounded away from zero.  Throughout this paper we will use the notation 
\begin{align*}
a \wedge b &:= \left(\begin{array}{cccc} \min(a_{11},b_{11}) & \min(a_{12},b_{12}) & \dots & \min(a_{1 d_2},b_{1 d_2}) \\
                                         \min(a_{21},b_{21}) & \min(a_{22},b_{22}) & \dots & \min(a_{2 d_2},b_{2 d_2}) \\
                                            \vdots           &    \vdots           & \ddots &   \vdots                   \\ 
                                         \min(a_{d_1 1},b_{d_1 1}) & \min(a_{d_1 2},b_{d_1 2}) & \dots & \min(a_{d_1d_2},b_{d_1d_2}) \end{array}\right),\\
a \vee b &:= \left(\begin{array}{cccc} \max(a_{11},b_{11}) & \max(a_{12},b_{12}) & \dots & \max(a_{1 d_2},b_{1 d_2}) \\
                                       \max(a_{21},b_{21}) & \max(a_{22},b_{22}) & \dots & \max(a_{2 d_2},b_{2 d_2}) \\
                                          \vdots           &    \vdots           & \ddots &   \vdots                   \\ 
                                       \max(a_{d_1 1},b_{d_1 1}) & \max(a_{d_1 2},b_{d_1 2}) & \dots & \max(a_{d_1d_2},b_{d_1d_2}) \end{array}\right)
\end{align*} 
where $a,b \in \bbr^{d_1 \times d_2}$ for some $d_1,d_2 \in \bbn$.  Additionally let $a^+ :=  a \vee 0$ and $a^- := (-a) \vee 0$ where $a \in \bbr^d$ for some $d \in \bbn$.

As described in \cite{EN01}, any financial agent $i \in \{1,2,...,n\}$ may be a creditor or obligor to other agents.  However, in contrast to \cite{EN01}, we consider these liabilities in multiple currencies that must be fulfilled in the physical assets rather than some num\'{e}raire.  This distinction, while unimportant in frictionless markets, will be necessary in Section~\ref{Sec:mAsset} when considering the impact on prices caused by the transactions undertaken by the firms.  Let $L_{ij}^k \geq 0$ be the contractual obligation of firm $i$ towards firm $j$ in asset $k$.  Further, we assume that no firm has an obligation to itself in any asset, i.e., $L_{ii}^k = 0$.  The \emph{total liabilities} of agent $i$ in asset $k$ are given by
\[\bar p_i^k := \sum_{j = 1}^n L_{ij}^k.\]
We can define the vector $\bar p^k \in \bbr^n_+$ as the vector of total obligations of each firm in asset $k$.
On the asset side of the balance sheet, each firm $i = 1,2,...,n$ has an initial endowment of $x_i^k \geq 0$ in each $k = 1,2,...,m$ asset.  We refer to Figure~\ref{Fig:bs-book} for a visual representation of the stylized book value of assets and liabilities for a representative firm with $m = 2$ assets and with market prices $q \in \bbr^m_{++}$.  Though firms may alter their borrowing based on market prices due to, e.g., new monetary policy in response to altered exchange rates, modifying the balance sheet in such a way is outside the scope of the current work.  We refer to \cite{SSB16b,SSB17,BF18} for consideration of contingent liabilities in the single ($m = 1$) asset setting.

The \emph{relative liabilities} of firm $i$ to firm $j$ in asset $k$, i.e., the fractional amount of total liabilities of firm $i$ towards firm $j$ in asset $k$, are given by 
\[a_{ij}^k = \begin{cases}\frac{L_{ij}^k}{\bar p_i^k} & \text{if } \bar p_i^k > 0 \\ \frac{1}{n} & \text{if } \bar p_i^k = 0\end{cases}.\]
We define the matrices $A^k = (a_{ij}^k)_{i,j = 1,2,...,n}$ with the property (by construction) $\sum_{j = 1}^n a_{ij}^k = 1$ for any $i$ and $k$.  In the case that $\bar p_i^k = 0$ we are able to choose $a_{ij}^k$ arbitrarily, we let $a_{ij}^k = \frac{1}{n}$ in that case so that the summation is equal to $1$.  Any financial firm may default on their obligations in asset $k$ if they do not hold a sufficient number of that asset.  We assume, as per \cite{EN01}, that in case of default the realized payments will be made in proportion to the size of the obligations, i.e., based on the relative liabilities matrix $A^k$ and without prioritization of payments to any firm.
That is, the realized value (in physical units) of firm $i$'s interbank assets in asset $k$ is given by
\[\sum_{j = 1}^n a_{ji}^k [\bar p_j^k \wedge y_j^k]\]
when firm $j \neq i$ holds $y_j^k$ units of asset $k$.  Encoded in this equation is the notion that if firm $j$ has more assets than liabilities in asset $k$ then it will pay out in full ($a_{ji}^k \bar p_j^k = L_{ji}^k$), otherwise it will pay out its holdings proportionally to what it owes.  This realized balance sheet is depicted in Figure~\ref{Fig:bs-balance} with $m = 2$ assets and with market prices $q \in \bbr^m_{++}$.

As we consider the setting in which all liabilities must be paid in physical assets, we need to consider an additional step to find the realized holdings for each bank in the system.  For instance, Figure~\ref{Fig:bs-balance} depicts the firm with positive mark-to-market capital, but a deficit in the first asset.  Thus, as depicted in Figure~\ref{Fig:bs-holdings}, they would have to transfer some units of the second asset so as to cover this liability.  As Figure~\ref{Fig:bs-holdings} considers the frictionless market, the realized capital for the firm before and after the transaction will remain constant, and as such this system is functionally equivalent to (a generalization of) the payment model from \cite{EN01}.  However, if price impacts were introduced (see Section~\ref{Sec:mAsset}) then more complicated firm behavior needs to be considered and a reduction to mark-to-market values is insufficient to describe the entire system.  The details of the firm behavior through a utility maximization problem is provided in the next section.

\section{The model}\label{Sec:model}
In this section we will first introduce the clearing framework for multi-currency obligations without price impacts.  In this setting we provide results on existence and uniqueness, which generalizes those results from~\cite{EN01,E07}.  In this case we are able to consider a fictitious default algorithm as was first considered in~\cite{EN01}.  The framework without price impacts is of interest because it is mathematically tractable.  Additionally, from a financial perspective it is of interest due to the generality of the payment schemes provided herein as well as allowing for clear heterogeneous shocks to the various institutions.  Under such a setting the use of multi-layered networks is unnecessary as an approach with appropriate prioritization of payments as in~\cite{E07} can be taken instead on the marked-to-market wealth.  However, with these results, we introduce price impacts due to the transfer of assets undertaken by the firms.  These market impacts cannot be wholly described with only the marked-to-market wealths and thus require the use of vector-valued, i.e.\ multi-layered, networks.  We conclude this section by considering the resultant equilibrium exchange rates achieved after an initial shock to the asset values.  This allows us to classify when the system of banks exacerbates, and when it mitigates, the effects of a financial crisis.  

We wish to compare this model with prior notions of fire sales in the~\cite{EN01} framework, e.g.~\cite{AW_15}.  In such works, all obligations are denominated in the same (cash) asset and illiquid assets are sold at a discount in order to cover these cash shortfalls.  By taking such an approach, the monotonicity of the clearing mechanism is immediate and Tarski's fixed point provides existence of clearing payments and prices.  However, in this work, banks are given freedom to both buy and sell assets so as to cover their obligations (in multiple assets) and to, for instance, purchase assets at a discount so as to increase their utility.  The existence of clearing prices and portfolio holdings requires more thorough comparative static results (that are provided in the appendix), and ultimately does not result in a lattice of equilibrium solutions.

\subsection{Financial contagion without market impacts}\label{Sec:monotonic}

Fix the behavior of all firms but firm $i$, i.e., the amount of each asset that all firms but $i$ hold is $y_{-i} \in \bbr^{(n-1)\times m}_+$ (with firm $j$ holding $y_j \in \bbr^m_+$), and the relative prices is given by the vector $q \in \bbr^m_{++}$.  The amount of each asset that firm $i$ has immediately available due to the payments from the other firms is given by
\[\left(x_i^k + \sum_{j = 1}^n a_{ji}^k \left[\bar p_j^k \wedge y_j^{k}\right]\right)_{k = 1,2,...,m}.\]
As described in \cite{EN01}, and depicted in Figure~\ref{Fig:bs-balance}, firms have available the sum of the endowment $x_i^k$ and the realized interbank assets $\sum_{j = 1}^n a_{ji}^k [\bar p_j^k \wedge y_j^{k}]$.
Following the concept of limited liabilities (i.e., no firm pays more than it owes) and absolute priority (i.e., no firm accumulates positive equity until all debts are paid in full), the holdings of firm $i$ are such that 
\[\left(\exists k^* \in \{1,2,...,m\}: \; y_i^{k^*} > \bar p_i^{k^*}\right) \quad \Rightarrow \quad y_i \geq \bar p_i.\]
We assume that additional regulatory rules apply to the multi-asset payments.  That is, regulators may enforce, e.g., a prioritized payment (as in, e.g., \cite{E07}) or pro-rata payment (as in, e.g., \cite{EN01}) between different assets or currencies.  These rules are encoded in some monotonic, strictly concave, and supermodular \emph{payment utility function} $h_i$.  The payments made by firm $i$ are given by
\begin{equation} 
\label{Eq:payment-h} P_i(y,q) = \argmax_{p_i \in [0,\bar p_i]} \left\{h_i(p_i;y_{-i},q) \; \left| \; \sum_{k = 1}^m q_k p_i^k \leq \sum_{k = 1}^m q_k \left(x_i^k + \sum_{j = 1}^n a_{ji}^k \left[\bar p_j^k \wedge y_j^{k}\right]\right)\right.\right\}.
\end{equation}
The payment function $P_i$ is defined, given the portfolio holdings of all other firms and the prevailing market price, so that the mark-to-market value of the payments does not exceed the available marked-to-market realized assets.  Additionally, by constraining the payments between $0$ and $\bar p_i$, we enforce the limited liabilities assumption.  The inclusion of the payment utility function $h_i$ is to guarantee that the resultant payments will satisfy the desired regulatory environment (e.g., prioritization or proportionality).  We refer to Section~\ref{Sec:h} for constructions of the payment utility function under financially interesting regulations.  This payment scheme is general enough to cover regulatory environments beyond the standard frameworks in the literature, i.e.\ proportional and prioritized payments, to include, e.g., a surplus repayment scheme described in Example~\ref{Ex:surplus}.

However, firm $i$ may choose to trade more assets than required to make its payments; this additional trading will be done in order to optimize some \emph{utility function} $u_i$.  To guarantee absolute priority, the final number of assets that firm $i$ holds must exceed, in each asset, the payment $P_i(y,q)$.  In this way the utility function is redundant, and unnecessary, for firms that are insolvent as they must cover exactly their payments $P_i(y,q)$.  Further, we constrain the actions of each bank so that it can obtain its desired portfolio without loss of mark-to-market valuation from its marked-to-market assets.  Thus the vector of asset holdings for firm $i$ is given by the bilevel program
\begin{equation} \label{Eq:holding}
y_i \in Y_i(y,q) = \argmax_{e_i \in \bbr^m_+} \left\{u_i(e_i;y_{-i},q) \; \left| \; \begin{array}{rcl}e_i &\geq& P_i(y,q),\\ \sum_{k = 1}^m q_k e_i^k &\leq& \sum_{k = 1}^m q_k \left(x_i^k + \sum_{j = 1}^n a_{ji}^k \left[\bar p_j^k \wedge y_j^{*k}\right]\right) \end{array}\right.\right\}.
\end{equation}
By enforcing the non-negativity constraint, we encode a no-short selling assumption.  Note that we allow firms to throw away wealth in determining their final portfolio holdings.  While mathematically this is possible, all examples considered herein will guarantee that any value in $Y_i(y,q)$ will have terminal (mark-to-market) wealth equal to the value of the firm's assets.

With the given rules for repayment and firm behaviors, we are able to fully describe the clearing mechanism for asset holdings.  Given an asset holding matrix $y \in \bbr^{n \times m}_+$ and pricing vector $q \in \bbr^m_{++}$ the updated asset holdings are given by the \emph{clearing mechanism} $Y$  
where $(Y_i)_{i = 1,2,...,n}$ is given in~\eqref{Eq:holding}.  Implicitly within this clearing mechanism, the regulatory agency has a role to play by specifying the payment utility function which determines the payment function $P_i$ (defined in~\eqref{Eq:payment-h}) for each firm $i$.  

We use the clearing mechanism to compute the \emph{realized holdings} $y(q) \in \bbr^{n \times m}_+$ under the pricing vector $q \in \bbr^m_{++}$.  This is provided by a fixed point of the clearing mechanism, i.e.,
\[y(q) \in Y(y(q),q).\]
We now consider conditions for the existence of maximal and minimal clearing solutions $y(q)$, which is the general property satisfied in the Eisenberg-Noe model, under a crisis price of $q$.  These results are then used to prove a sufficient condition for the uniqueness of the clearing solution by guaranteeing that the maximal and minimal solutions must coincide.  Note that, due to the generality of the payment scheme, encoded by the payment utility functions, it is not possible to directly apply the results of \cite{EN01} on the marked-to-market assets and liabilities for each firm; however, in the special case discussed in Remark~\ref{Rem:proportional-unique} below this approach could be taken.  Additionally, due to the utility maximizing behavior of the regulators (through the payment utility function) and solvent banks (through the utility functions), we must consider comparative statics of the bilevel optimization problems for each firm to prove the result as provided in the appendix.

\begin{lemma}\label{Lemma:monotonic}
Fix a price $q \in \bbr^{m}_{++}$.
Let the \emph{payment utility functions} $h_i: [0,\bar p_i] \times \bbr^{(n-1) \times m}_+ \times \bbr^m_{++} \to \bbr$ be strictly increasing, strictly concave, and supermodular in its first argument.
Let the \emph{utility functions} $u_i: \bbr^m_+ \times \bbr^{(n-1) \times m}_+ \times \bbr^m_{++} \to \bbr$ be concave and supermodular in its first argument.  Additionally assume that $Y_i(y,q)$ (defined in~\eqref{Eq:holding}) is singleton-valued for any $y \in \bbr^{n \times m}$ and for all agents $i$. 
\begin{enumerate}
\item\label{Lemma:monotonic-exist} There exists a greatest and least clearing holdings $y^{\uparrow}(q) \geq y^{\downarrow}(q)$ satisfying $y = Y(y,q)$.
\item\label{Lemma:monotonic-equity} The positive equity of all firms is equal for every fixed point, i.e., $(y_i^{\uparrow}(q) - \bar p_i)^+ = (y_i^{\downarrow}(q) - \bar p_i)^+$ for every firm $i$.
\end{enumerate}
\end{lemma}

These results on a greatest and least clearing portfolio holdings generalize Theorem 1 of \cite{EN01}.  In fact in the $m = 1$ asset case, all payment utility functions and utility functions satisfying the conditions of Lemma~\ref{Lemma:monotonic} (for example setting $h_i(p_i;y_{-i},q) := p_i^2$ and $u_i(e_i;y_{-i},q) := e_i$) recover exactly the Eisenberg-Noe payments and assets as given by Theorem 1 of \cite{EN01}.

In the following corollary, we introduce an additional node to the financial system.  Denoted as node $0$, this ``firm'' represents all institutions and persons not included in the system of $n$ banks.  This notion is developed in more detail in, e.g.,~\cite{GY14,feinstein2014measures}.  In particular, we will assume that this ``societal node'' acts as a sink to the system, i.e., it has no obligations into the network.  This is incorporated in the assumption that the societal node will never default on its obligations as the initial endowments come from outside the original system.  If a default of node $0$ were desired, this could be included by stressing the initial endowments of the $n$ firms. 
\begin{corollary}\label{Cor:monotonic-unique}
Consider the setting of Lemma~\ref{Lemma:monotonic}.  If $L_{i0}^k > 0$ and $L_{0i}^k = 0$ for every firm $i$ and asset $k$ then the equilibrium holdings under price $q \in \bbr^m_{++}$ is unique, i.e., $y^*(q) := y^{\uparrow}(q) = y^{\downarrow}(q)$.  
\end{corollary}

\begin{remark}\label{Rem:proportional-unique}
Under specific choices of payment utility and utility functions $h_i$ and $u_i$ satisfying the conditions of Lemma~\ref{Lemma:monotonic}, we can give weaker conditions for uniqueness.  For instance, under proportional transfers (see Example~\ref{Ex:priority-proportional} with $\mu = 0$) with minimal trading (see Example~\ref{Ex:min-trading}) if $(q^{*\T}x,\sum_{k = 1}^m q_k^* L^k)$ is a regular network in the setting of \cite{EN01} (i.e., all firms have a directed path, possibly of length 0, to a firm with positive endowment, see Definition 5 of \cite{EN01}) then the clearing holdings are unique.
\end{remark}

We will introduce a modified version of the \emph{fictitious default algorithm} from \cite{EN01,RV13,AW_15,AFM16,feinstein2015illiquid} for the construction of the greatest portfolio holdings $y^{\uparrow}(q)$ under price $q \in \bbr^m_{++}$.  In particular, as with the prior fictitious default algorithms, this algorithm will converge after at most $n$ iterations since the set of defaulting banks is monotonic.  Though this algorithm converges within the finite number of iterations, it includes a fixed point problem in each iteration as is also the case in, e.g., \cite{AW_15,AFM16}.
\begin{alg}\label{Alg:clearing}
Consider the setting of Lemma~\ref{Lemma:monotonic} such that, additionally, 
\begin{align*}
h_i(p_i;y_{-i},q) &= h_i(p_i;\bar p_{-i} \wedge y_{-i},q)\\
u_i(e_i;y_{-i},q) &= u_i(e_i;\bar p_{-i} \wedge y_{-i},q)
\end{align*}
for every firm $i$ and every $p_i \in [0,\bar p_i]$, $e_i \in \bbr^m_+$, $y_{-i} \in \bbr^{(n-1)\times m}_+$, and $q \in [\underline q,\overline q]$. The greatest portfolio holdings $y^{\uparrow}(q)$ under price $q \in \bbr^m_{++}$ can be found by the following algorithm in at most $n$ iterations of the following.  
Initialize $\alpha = 0$, $p^\alpha = \bar p$, and $D^\alpha = \emptyset$.  Repeat until convergence:
\begin{enumerate}
\item Increment $\alpha = \alpha+1$;
\item For any firm $i = 1,2,...,n$ and asset $k = 1,2,...,m$, define the portfolio holdings by
$y_{ik}^\alpha = x_{ik} + \sum_{j = 1}^n a_{ji}^k p_{jk}^{\alpha-1};$
\item Denote the set of insolvent banks by
$D^\alpha := \left\{i \in \{1,2,...,n\} \; | \; q^{\T}(y_i^\alpha - \bar p) < 0\right\}$;
\item\label{Alg:terminate} If $D^\alpha = D^{\alpha-1}$ then exit loop;
\item Define the matrix $\Lambda^\alpha \in \{0,1\}^{n \times n}$ so that
$\Lambda_{ij}^\alpha = \begin{cases}1 &\text{if } i = j \in D^\alpha \\ 0 &\text{else}\end{cases}$.
$p^\alpha = \hat p$ is the maximal solution to the following fixed point problem
\begin{align}
\label{Eq:alg-fixedpt} \hat p &= \left(I - \Lambda^\alpha \right)\bar p + \Lambda^\alpha P\left((I - \Lambda^\alpha)\bar p + \Lambda^\alpha \hat p,q\right).
\end{align}
\end{enumerate}
After terminating the loop the clearing holdings can be computed by $y = Y(p^\alpha,q)$.
\end{alg}
The additional condition required for Algorithm~\ref{Alg:clearing} for the payment utility and utility functions states firm $i$ determines how much it pays or holds based only on the payments of the other firms $\bar p_{-i} \wedge y_{-i}$ and \emph{not} on the actualized holdings of the other firms $y_{-i}$.

We will finish our discussion of the equilibrium portfolio holdings without price impacts by considering a simple two bank example for which the clearing solution can be computed analytically.  We will refer back to this example at the end of the section on price impacts and attained equilibria as well.
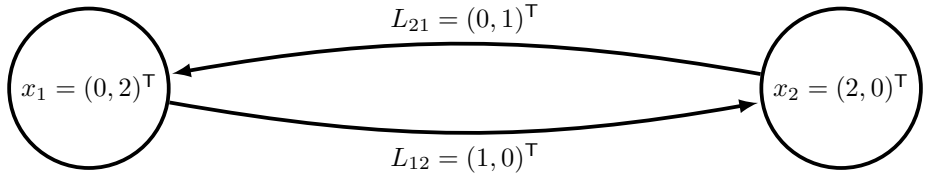
\begin{figure}[h!]
\centering
\begin{tikzpicture}
\tikzset{node style/.style={state, minimum width=0.36in, line width=0.5mm, align=center}}
\node[node style] at (0,0) (x1) {$x_1 = (0,2)^\T$};
\node[node style] at (10,0) (x2) {$x_2 = (2,0)^\T$};
\draw[every loop, auto=right, line width=0.5mm, >=latex]
(x1) edge[bend right=10] node {$L_{12} = (1,0)^\T$} (x2)
(x2) edge[bend right=10] node {$L_{21} = (0,1)^\T$} (x1);
\end{tikzpicture}
\caption{Example~\ref{Ex:no-impact}: A graphical representation of the network model with 2 banks and 2 assets which accepts more than one clearing solution.}
\label{Fig:nonunique}
\end{figure}
\begin{example}\label{Ex:no-impact}
Consider the network with $n = 2$ banks and $m = 2$ assets depicted in Figure~\ref{Fig:nonunique}.  That is, the first institution holds 2 units of the second asset and owes 1 in the first asset to the second institution; vice versa the second institution holds 2 units of the first assets and owes 1 in the second asset to the first institution.  Note that any choice of $(h_i)_{i = 1,2}$ satisfying the conditions of Lemma~\ref{Lemma:monotonic} is equivalent in equilibrium since, for both banks, all obligations are only in single assets.  Consider a utility function $u_i$ that minimizes the total amount of trading in the market (see Example~\ref{Ex:min-trading} below for more details).  Without loss of generality we will let asset $1$ denote the num\'{e}raire asset (i.e., $q_1 = 1$ throughout this example).  As discussed in Remark~\ref{Rem:proportional-unique}, for any price $q_2 > 0$, this system will have a \emph{unique} clearing solution given by:
\begin{align*}
y_1^*(q) &= \left(\begin{array}{c} \min\left(1,3q_2\right) \\ \left(2 + \min\left(1,3/q_2\right) - 1/q_2\right)^+ \end{array}\right) \qquad \text{ and } \qquad y_2^*(q) = \left(\begin{array}{c} \left(2 + \min\left(1,3q_2\right) - q_2\right)^+ \\ \min\left(1,3/q_2\right) \end{array}\right).
\end{align*}
\end{example}

\subsection{Financial contagion with market impacts}\label{Sec:mAsset}
The results on the clearing portfolio holdings without market impacts generalize the results of~\cite{EN01}.  In fact in the $m = 1$ asset case, all payment utility functions and utility functions satisfying the conditions of Lemma~\ref{Lemma:monotonic} (for example setting $h_i(p_i;y_{-i},q) := p_i^2$ and $u_i(e_i;y_{-i},q) := e_i$) recover exactly the Eisenberg-Noe payments and assets as given by Theorem 1 of \cite{EN01} if there are no market impacts.  However, market impacts (due to asset transfers undertaken by the firms) introduce further feedback effects on the firms and cannot be considered in the $m = 1$ scheme from \cite{EN01}.
In this section, we will first introduce the inverse demand functions, which we use to model the market impacts of firm behavior.  We will then use this model of market impacts to consider the existence of clearing prices and portfolio holdings, thus generalizing the prior section.  

The price of the assets is given by a vector valued \emph{inverse demand function} $F: \bbr^m \to [\underline q,\overline q] \subseteq \bbr^m_+$ for minimum and maximum prices $\underline q = (1 , \underline q_2 , \dots , \underline q_m)^\T$ and $\overline q = (1 , \overline q_2 , \dots , \overline q_m)^\T$ where the first asset is the num\'{e}raire.  
The inverse demand function maps the quantity of each asset to be sold into a price per unit in the num\'{e}raire.  
The liquidation value, in the num\'{e}raire, of the portfolio $z \in \bbr^m$ is thus given by $z^\T F(z)$.  We will impose the following assumption for the remainder of this paper.
\begin{assumption}\label{Ass:idf}
The inverse demand function $F: \bbr^m \to [\underline q,\overline q] \subseteq \bbr^m_{++}$ is continuous and nonincreasing.  
For simplicity, the inverse demand function has the form $F(z) := (1 , f_2(z_2) , \dots , f_m(z_m))^\T$ for every $z \in \bbr^m$.
\end{assumption}
\begin{remark}\label{Rem:numeraire}
In the construction of the inverse demand function, we choose to take the first asset to be the num\'{e}raire asset.  However, this assumption need not be made for the results of this work.  In fact, a fictitious num\'{e}raire asset can be chosen instead and thus we would choose $\underline q_1 < \overline q_1$ and some function $f_1(z_1)$ would be necessary as well.  Further, the assumption that no cross impacts exist in the pricing can be eliminated without affecting the results of this work so long as the inverse demand function is continuous and nonincreasing.
\end{remark}

\begin{remark}\label{Rem:idf}
Rather than introduce a num\'{e}raire, we can substitute the bid-ask matrix (see, e.g., \cite{S04}) with components $\pi_{k_1k_2}(z) := \frac{F_{k_2}(z)}{F_{k_1}(z)}$ for the inverse demand function.  
Equivalently, from the bid-ask matrix $\Pi: \bbr^m \to \bbr^{m \times m}_{++}$, we can construct a set of inverse demand functions such that the first asset is the num\'{e}raire asset by defining $F_k(z) := \pi_{1k}(z)$.
Similarly, rather than introduce the inverse demand function, we could consider the demand curve for the nonbanking sector (as done in, e.g.,~\cite{CL15}).  We consider the inverse demand function since it simplifies the formulations of this paper, though it can be constructed from the demand curve of the nonbanking sector.
\end{remark}

Additionally, though not needed for the results of this paper, we will generally assume that the inverse demand function satisfies the condition that $z \in \bbr^m \mapsto z^\T F(z)$ is a strictly increasing mapping.  That is, the liquidation value of a portfolio is strictly increasing as portfolio holdings get larger.  Note that we consider this for portfolios with short positions, i.e., $z_k < 0$, as well.
\begin{remark}\label{Rem:idf2asset}
In the case when there are $m = 2$ assets, following Assumption~\ref{Ass:idf}, we will assume throughout that $F_1 \equiv 1$ and $F_2(z) := f_2(z_2)$ for every $z \in \bbr^2$ for some continuous and nonincreasing inverse demand function $f_2: \bbr \to [\underline q_2,\overline q_2]$.  That is, the first asset will act as the num\'{e}raire asset and the price of the second asset will depend only on the number of units being bought or sold in that asset.  
In prior works, e.g.~\cite{CFS05,AFM16,feinstein2015illiquid}, the inverse demand function is only defined as being a function of non-negative units being sold; using a symmetric argument, we can define an inverse demand function on the entire real line from this half-line inverse demand function.  Consider $\hat f: \bbr_+ \to [\underline q_2,1]$ to be a continuous and nonincreasing inverse demand function such that $\alpha(z_2) := z_2 \hat f(z_2)$ is strictly increasing in $z_2 \in \bbr_+$.  Then we can define the full inverse demand function in a symmetric way as
\begin{equation}\label{Eq:idf2asset}
F_2(z) := \begin{cases}\hat f(z_2) &\text{if } z_2 \geq 0\\ \frac{1}{\hat f(\alpha^{-1}(-z_2))} &\text{if } z_2 < 0\end{cases}.
\end{equation}
The notion of symmetry is due to the fact that selling $z_2$ units of the second asset is equivalent to purchasing $\alpha(z_2)$ units of the first asset (i.e., selling $-\alpha(z_2)$ units).  Thus, when purchasing $|z_2|$ units of asset $2$, for $z_2 < 0$, we can consider selling $\alpha^{-1}(-z_2)$ units of the first asset.  With the added assumption that the first asset has the same inverse demand function $\hat f$ (when selling units of asset $1$ denominated in the second asset), and changing num\'{e}raire back to asset $1$, results in the inverse demand function as presented in~\eqref{Eq:idf2asset}.
\end{remark}

With the model of price impacts, given by the inverse demand function $F$, we want to return again to the clearing model for firm portfolio holdings.  Given an initial price $q_0 \in [\underline q,\overline q]$, Section~\ref{Sec:monotonic} provides the firm behavior $y^*(q_0)$.  However, if these sales are actualized, this leads to an updated price $q$ due to the market impact of the transfers undertaken by each firm.  In particular, the updated prices are a function of the net difference between what is initially available to each firm and the final holdings.  The initial shock $q_0 \in [\underline q,\overline q]$ would be generated by actions from agents outside our system.  That is, initially some quantity $\gamma_0 \in \bbr^m$ is transacted so that $q_0 = F(\gamma_0)$.  The clearing prices, subject to the initial shock $q_0$, are thus given by
\begin{align} 
\nonumber q &= F\left(\gamma_0 + \sum_{i = 1}^n \left(x_i^k + \sum_{j = 1}^n a_{ji}^k \left[\bar p_j^k \wedge y_j^{*k}(q)\right] - y_i^{*k}(q)\right)_{k = 1,2,...,m}\right)\\
\nonumber &= F\left(\gamma_0 + \sum_{i = 1}^n x_i + \sum_{j = 1}^n \left(\left[\sum_{i = 1}^n a_{ji}^k\right] \left[\bar p_j^k \wedge y_j^{*k}(q)\right]\right)_{k = 1,2,...,m} - \sum_{i = 1}^n y_i^*(q)\right)\\
\label{Eq:price} &= F\left(\gamma_0 + \sum_{i = 1}^n \left(x_i + \left[\bar p_i \wedge y_i^*(q)\right] - y_i^*(q)\right)\right).
\end{align}
With the feedback effects from the inverse demand function, there is the potential for increased contagion than the initial shock propagating through the no market impact case of Section~\ref{Sec:monotonic} or the single asset setting of \cite{EN01}.  However, if firms choose to purchase a distressed asset due to the decrease in price, it is possible that a mitigating feedback loop is instituted that will ultimately cause fewer defaults and improved exchange rates compared to the initial shock $q_0$.

For the ease of utilizing these results, we will now provide a summary of all assumptions for Corollary~\ref{Cor:existence}.  These are exactly those from Corollary~\ref{Cor:monotonic-unique}, Assumption~\ref{Ass:idf}, and assuming both types of utility functions are jointly continuous.
\begin{assumption}\label{Ass:existence}
Let the network be such that all firms have obligations to a societal node and no obligations from such a node in each asset, i.e., $L_{i0}^k > 0$ and $L_{0i}^k = 0$ for every firm $i$ and asset $k$.
Let the inverse demand function $F: \bbr^m \to [\underline q,\overline q]$ satisfy Assumption~\ref{Ass:idf}, i.e., be continuous and nonincreasing.
Let the payment utility functions $h_i: [0,\bar p_i] \times \bbr^{(n-1) \times m}_+ \times \bbr^m_{++} \to \bbr$ be strictly increasing, strictly concave, and supermodular in its first argument and jointly continuous for every bank $i$.
Let the utility functions $u_i: \bbr^m_+ \times \bbr^{(n-1) \times m}_+ \times \bbr^m_{++} \to \bbr$ be concave and supermodular in its first argument and jointly continuous for every bank $i$.  Additionally assume that $Y_i(y,q)$ (defined in~\eqref{Eq:holding}) is singleton-valued for any $y \in \bbr^{n \times m}$ and for all agents $i$. 
\end{assumption}

\begin{corollary}\label{Cor:existence}
Let Assumption~\ref{Ass:existence} hold.
Let $\gamma_0 \in \bbr^m$ be an initial set of transactions that result in a price shock $q_0 = F(\gamma_0) \in [\underline q,\overline q]$. There exists a fixed point price
\[q^* = F\left(\gamma_0 + \sum_{i = 1}^n (x_i + [\bar p_i \wedge y_i^*(q^*)] - y_i^*(q^*))\right)\]
and resultant portfolio holdings $y^*(q^*)$.
\end{corollary}
In comparison to prior works in the Eisenberg-Noe framework, the existence of the clearing prices does not follow from a monotonicity argument with Tarski's fixed point, but rather from Brouwer's fixed point theorem (as detailed in the appendix).

\begin{remark}
Assumption~\ref{Ass:existence} can be weakened and still guarantee existence of joint clearing portfolio holdings and prices.  In fact, there exists joint clearing holdings and prices so long as:
\begin{itemize}
\item the payment utility functions $h_i$ are jointly continuous and both strictly increasing and strictly quasi-concave in their first argument, 
\item the utility functions $u_i$ are jointly continuous and quasi-concave in their first argument, and 
\item the inverse demand function $F$ satisfies Assumption~\ref{Ass:idf}.
\end{itemize}
This can be proven using an iterated application of the Berge Maximum Theorem (see, e.g., Theorem 17.31 in~\cite{AB07}) for $P_i$ and $Y_i$ followed by an application of the Kakutani Fixed Point Theorem (see, e.g., Theorem 3.2.3 in \cite{AF90}) to attain the existence of a fixed point $(y^*,q^*)$.\fn{This statement is formalized and proven in Theorem 3.2 of an older preprint version of this text available at \url{https://arxiv.org/pdf/1702.07936v2.pdf}.}  In fact, using the logic of that proof, we can allow for \emph{continuous} admissible valuation functions (as defined in \cite{V17distress}) $\mathbb{V}_i^k(y_i^k/\bar p_i^k) \in [0,1]$ with payments $L_{ij}^k\mathbb{V}_i^k(y_i^k/\bar p_i^k)$ from firm $i$ to $j$ in asset $k$.  In this paper, we exclusively consider the continuous admissible valuation function $\mathbb{V}^{EN}(z) = 1 \wedge z^+$ for all firms and all assets.  We refer to \cite{V17distress} for further discussions on admissible valuation functions and the relation to bank distress and bankruptcy costs.  For simplicity, we focus on the stronger assumptions introduced in Assumption~\ref{Ass:existence} for the remainder of this paper.
\end{remark}

\begin{remark}
Note that we allow (and likely enforce) firms to both buy and sell assets during a fire sale, this is in contrast to earlier works such as \cite{CFS05,AFM16,feinstein2015illiquid,feinstein2016leverage}.  Such an approach is necessary to consider the cross-currency obligations exhibited in many systemic crises (e.g., \cite{DLM11}).  In such a setting, firms will transfer between the currencies or, more generally, assets in order to satisfy the different obligations.  That is, for instance, a firm in the United States may need to sell US dollars for euros in order to fulfill European liabilities, while a European firm may enact the reverse transaction within the same international financial system.  Further, we allow for solvent firms to use their excess wealth in order to maximize their utility.  By allowing this, the contagion effects of a fire sale could be partially mitigated if firms purchase an asset in a fire sale (or sell an asset being bought in excess).  The extreme mitigation in which the system has no net changes in currency holdings (i.e., all sales by one firm are purchased by a separate firm in the system) follows the model of \cite{EN01} with no illiquidity; this extreme mitigation occurs when there is a symmetry in endowments and obligations between the assets or currencies.  In contrast, the fire sales can have a large impact on the health of the various firms when there is an asymmetry in the system (e.g., between US dollars and Thai baht during the 1997 Asian financial crisis as discussed in \cite{ABK03}).  
\end{remark}

\subsection{Attained equilibrium in $m = 2$ asset case}\label{Sec:attained}
In general, uniqueness of the clearing holdings and prices is \emph{not} guaranteed.  We refer to Example~\ref{Ex:nonunique} below, which provides an illustration of multiple clearing solutions under the setting of Example~\ref{Ex:no-impact}, i.e., a simple two asset network.  However, though there may exist more than one clearing solution, the system can only attain a single equilibrium.  This section will focus on the t\^atonnement process by which an equilibrium is attained after the initial shock occurs in the $m = 2$ asset scenario.

Consider an initial price shock $q_0 \in [\underline q,\overline q]$ generated by asset transfers $\gamma_0 \in \bbr^m$.  Initially the firms would want to reach $y^*(q_0)$, but as they implement these transactions they will impact the prices in a \emph{continuous} way.  In particular, the prices will update along the direction of difference between the ``desired'' price and the current price, i.e., beginning from $q_0$
\begin{equation}\label{Eq:tatonnement}
dq_t = \left[F\left(\gamma_0 + \sum_{i = 1}^n \left(x_i + \left[\bar p_i \wedge y_i^*(q_t)\right] - y_i^*(q_t)\right)\right) - q_t\right] dt.
\end{equation}
The attained clearing solution is exactly the asymptotic solution of this process, if it exists.  This procedure is often called the t\^atonnement process in the economics literature (see, e.g.,~\cite{BFK14}).

We wish to emphasize that although the set of clearing solutions given a shock $q_0$ may not be a singleton, (assuming convergence) the t\^atonnement process can only reach a single clearing solution.  Importantly, because the set of clearing solutions is not unique, it will frequently be the case that the attained clearing solutions (as a function of the initial shock $q_0$) will be discontinuous.  These points of discontinuity match exactly those stresses in which a financial crisis becomes a systemic crises.  That is, if a marginal change in initial shock can cause a radically different clearing solution to be attained.  We refer the reader to Example~\ref{Ex:nonunique}, and in particular to Figure~\ref{Fig:nonunique-q}, for a demonstration of such an event.  This discontinuity is fundamentally a result of the nonuniqueness of equilibria.  As an initial stress grows too large, the system may not be able to sustain a financially stabilizing equilibrium anymore as firms enter insolvency and, therefore, a jump in equilibria occurs to a more extreme outcome.  Such events are systemic as they are purely a result of firm insolvency that cannot be captured by looking solely at the aggregate of the financial system. 

\begin{proposition}\label{Prop:tatonnement}
Consider the setting of Assumption~\ref{Ass:existence} with only $m = 2$ assets.  The t\^atonnement process~\eqref{Eq:tatonnement} will converge to a clearing solution.
\end{proposition}

This t\^atonnement process, ultimately, provides the attained clearing solution which includes market impacts from the individual firm behaviors.  As mentioned previously, these market impacts can have either mitigating or exacerbating effects which cannot be captured by the Eisenberg-Noe framework with exogenous shocks only.  As we will investigate in numerical case studies in Section~\ref{Sec:casestudy}, the choice of regulatory framework and utility functions will affect the clearing solutions.  As a rule of thumb, and as expected, small shocks may be ``absorbed'' by the system, while large shocks are likely to be exacerbated and may potentially drive the price to the upper or lower bound.

\begin{example}\label{Ex:nonunique}
Consider the network with $n = 2$ banks and $m = 2$ assets depicted in Figure~\ref{Fig:nonunique} with parameters considered in Example~\ref{Ex:no-impact}.  
Consider linear price impacts on the second asset, i.e., $F_2(z) = \underline q_2 \vee (1 - bz_2) \wedge \overline q_2$ for some lower bound $\underline q_2 < \frac{1}{3}$, upper bound $\overline q_2 > 3$ on the prices, and price impact $b \in (0,1)$.  The set of clearing prices, as a function of the initial shock size $q_0 \in [\underline q,\overline q]$, are given by
\begin{align*}
Q^*(q_0) &= \begin{cases} \{q^0(q_0)\} & \text{if } \underline q_2 \leq q_{0,2} < -b + 2\sqrt{b}\\
                          \{q^0(q_0) , q^\downarrow(q_0) , q^\uparrow(q_0)\} & \text{if } -b + 2\sqrt{b} \leq q_{0,2} \leq 2b + \frac{1}{3}\\
                          \{q^\uparrow(q_0)\} & \text{if } 2b + \frac{1}{3} < q_{0,2} < 3-\frac{2}{3}b\\
                          \{q^1(q_0)\} & \text{if } 3-\frac{2}{3}b \leq q_{0,2} \leq \overline q_2 \end{cases}
\end{align*}
where the candidate clearing prices are provided by
\begin{align*}
q^\uparrow(q_0) &= \left(1,\frac{1}{2}\left[q_{0,2} + b + \sqrt{q_{0,2}^2 + 2b(q_{0,2}-2) + b^2}\right]\right)^\T\\
q^\downarrow(q_0) &= \left(1,\frac{1}{2}\left[q_{0,2} + b - \sqrt{q_{0,2}^2 + 2b(q_{0,2}-2) + b^2}\right]\right)^\T\\
q^0(q_0) &= \left(1,\left[q_{0,2} - 2b\right] \vee \underline q_2\right)^\T\\
q^1(q_0) &= \left(1,\frac{1}{2}\left[q_{0,2} + \sqrt{q_{0,2}^2 + 8b}\right] \wedge \overline q_2\right)^\T.
\end{align*}
We wish to note that $q^\uparrow((1,3-\frac{2}{3}b)^\T) = q^1((1,3-\frac{2}{3}b)^\T)$, so only one clearing solution exists at $q_{0,2} = 3-\frac{2}{3}b$.  Though multiple clearing prices exist for $q_{0,2} \in [-b + 2\sqrt{b},2b + \frac{1}{3}]$, a \emph{unique} price is attained using the t\^atonnement process.  This selector from the set of all clearing prices can be determined as a function of the initial shock $q_0 \in [\underline q,\overline q]$ to be given by
\[
q^*(q_0) = \begin{cases} q^0(q_0) & \text{if } \underline q_2 \leq q_{0,2} < -b + 2\sqrt{b}\\
                         q^\uparrow(q_0) & \text{if } -b + 2\sqrt{b} \leq q_{0,2} < 3-\frac{2}{3}b\\
                         q^1(q_0) & \text{if } 3-\frac{2}{3}b \leq q_{0,2} \leq \overline q_2 \end{cases}.
\]
Of particular note, $q^*$ is discontinuous at $q_{0,2} = -b + 2\sqrt{b}$ in general.  This provides the important notion that the system can, roughly, absorb a shock of size $-b + 2\sqrt{b}$ (with some exacerbating tendencies), but any shock larger than that will cause a near complete collapse of the system.  We refer the reader to Figure~\ref{Fig:nonunique-q} which displays the set of clearing prices $Q^*$ and the attained clearing price $q^*$ as functions of the initial shock $q_0$ where $\underline q_2 = 0.05$, $\overline q_2 = 5$, and $b = \frac{3}{8}$.  Both the full region of responses and a consideration of only downward shocks in the second asset price are presented.  Notably, at $q_{0,2} = -b + 2\sqrt{b} \approx 0.85$ the attained equilibrium price drops from roughly $0.6125$ to $0.10$.  As such, this simple system can be viewed as being able to withstand a 15\% drop in asset prices, but no more.
\begin{figure}
\centering
\begin{subfigure}[b]{.47\textwidth}
\includegraphics[width=\textwidth]{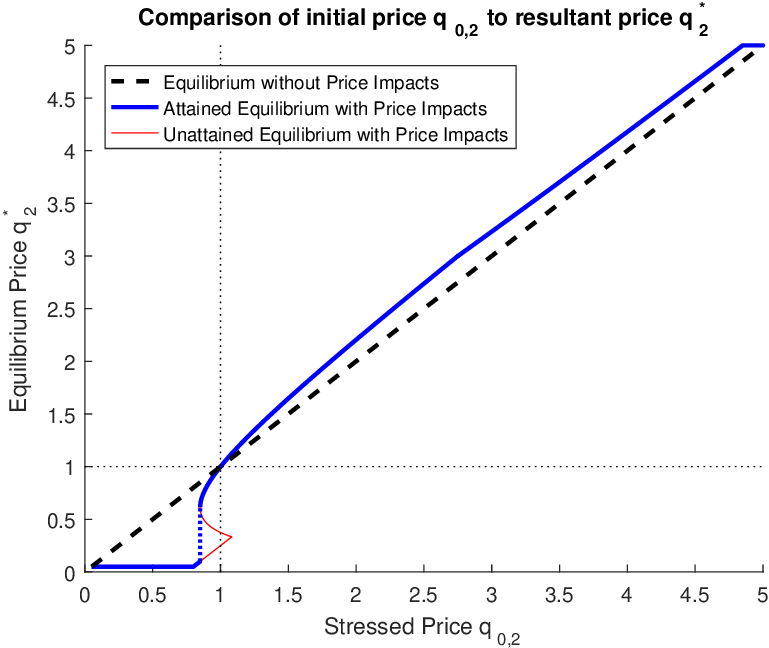}
\caption{The graph of all equilibrium prices subject to any initial shock $q_{0,2} \in [0.05,5]$.}
\label{Fig:analytical_full}
\end{subfigure}
~
\begin{subfigure}[b]{.47\textwidth}
\includegraphics[width=\textwidth]{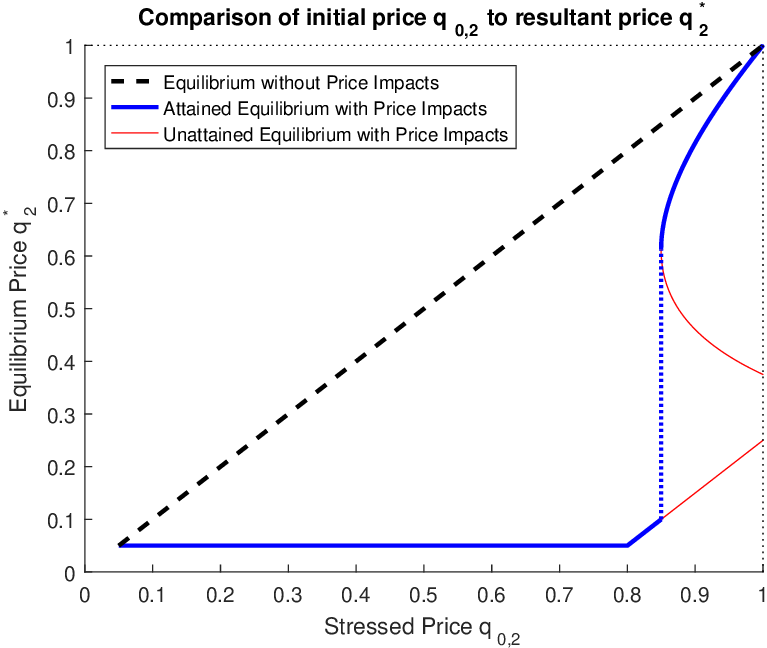}
\caption{The graph of all equilibrium prices subject to any downward shock to the second asset $q_{0,2} \in [0.05,1]$.}
\label{Fig:analytical_zoom}
\end{subfigure}
\caption{Example~\ref{Ex:nonunique}: A comparison of the set of equilibrium prices $Q^*$ to initial shocks $q_0$ with the attained equilibrium $q^*$ highlighted.}
\label{Fig:nonunique-q}
\end{figure}
\end{example}

\section{Example payment utility and utility functions}\label{Sec:examples}

In this section we present two possible choices for the payment utility function $h_i$ and three choices for the utility functions $u_i$ which satisfy Assumption~\ref{Ass:existence} and the additional conditions of Algorithm~\ref{Alg:clearing}.  For the payment utility function $h_i$ we present quadratic formulations that correspond to:
\begin{itemize}
\item \emph{Surplus transfers}: a firm only transfers from one asset to another if there is a surplus to exchange.
\item \emph{Prioritization with proportional payments}: a firm transfers all wealth to asset 1 first until that obligation is paid off in full, then attempting to fulfill obligations in the second asset, and so on through the $\mu^\text{th}$ asset, then attempting to fulfill all other obligations paying out in proportion to the total liabilities. As special cases, this setting includes both asset prioritization (i.e., all assets are ordered and given a strict seniority structure for repayment) and proportional payments (i.e., the amount of obligations filled follows the same proportion as the total liabilities).
\end{itemize}
For the utility function $u_i$ we present three options with clear meaning:
\begin{itemize}
\item \emph{Minimal trading}: a utility function such that firms choose to see how markets respond in the immediate aftermath of the crisis in order to determine their investment response, i.e., firms seek to minimize the total amount of trading between assets (once the rules to find the payments $P_i$ are taken into account).
\item \emph{Asset maximizing}: a utility function encoding a flight-to-quality which seeks to maximize the total number of units of a specific asset, at the expense of all other assets.  
\item \emph{Value maximizing}: a utility function which is given by the value of the final portfolio holdings of the firm since firms typically trade assets in order to maximize return on equity.  In particular, this utility function attempts to maximize the total pre-crisis wealth for a firm.  
\end{itemize}

\subsection{Sample payment utility functions}\label{Sec:h}
Here we will consider the details of two possible, meaningful, options for the choice of payment utility functions $h_i$ in~\eqref{Eq:payment-h}.  These are a surplus transfer rule and a prioritization of the first $\mu$ assets and proportional payments for the remainder rule.  Both of these sample payment utility functions satisfy Assumption~\ref{Ass:existence} and the conditions of Algorithm~\ref{Alg:clearing}.

\begin{example}\label{Ex:surplus}
Consider a regulatory framework in which a firm is only forced to transfer assets if there is a surplus that is not being used to cover obligations already.  
In an international financial system, such a regulatory framework would naturally exist if each (independent) regulatory body places priority on its own currency.  Any institution operating in multiple nations would be forced to follow the local regulations with their locally held endowments.  However, once a firm has satisfied all obligations in a currency, the regulatory requirements therein have been satisfied and they may exchange the surplus to any other currency still in deficit.
One possibility to describe this framework is represented mathematically by the quadratic payment utility function $h_i: [0,\bar p_i] \times \bbr^{(n-1)\times m}_+ \times [\underline q,\overline q] \to \bbr$ defined by:
\begin{align*}
h_i(p_i;y_{-i},q) &:= -\frac{1}{2}\left(c_i - p_i\right)^\T \diag\left(\left[\frac{q_1}{c_i^1-e_i^1},\frac{q_2}{c_i^2-e_i^2},\dots,\frac{q_m}{c_i^m-e_i^m}\right]^\T\right) \left(c_i - p_i\right)\\
c_i &= \bar p_i \vee e_i + \delta\\
e_i^k &= x_i^k + \sum_{j = 1}^n a_{ji}^k \left[\bar p_j^k \wedge y_j^{k}\right] \quad\quad (k = 1,2,...,m).
\end{align*}
The $\delta \in \bbr^m_{++}$ term that appears is to shift the center of the ellipsoidal level sets of $h_i$ rightward and upward from the maximum between the amount owed $\bar p_i$ and the amount held (pre-transfers) through market clearing $e_i$.  The $\delta$ is introduced solely to avoid a division by 0 in this representation of the surplus payment utility function.  In fact, this payment utility function is chosen such that the level sets are ellipsoids with center above both $\bar p_i$ and $e_i$, and such that the gradient of $h_i$ is $q$ at $e_i$.
\end{example}

\begin{example}\label{Ex:priority-proportional}
Consider a regulatory framework in which a prioritization scheme is applied to the first $\mu \in \{0,1,...,m\}$ assets, and all other assets are treated in equal proportion after those first $\mu$ assets are paid in full.  In the special cases that $\mu = 0$ this is a purely proportional payments regulation scheme, i.e., pro-rata.  In the case that $\mu = m$ this is a purely prioritized payments regulation scheme, i.e., a seniority structure as in \cite{E07}.  This may arise if asset $1$ is the local currency due to regulations favoring those payments. 
Financial institutions will pay off their balance in asset $1$ (including by transferring funds from all other assets), and only after that obligation is fulfilled will they begin filling asset $2$ and so on down the line until they pay off asset $\mu$.  Only after asset $\mu$ is paid off in full will the other obligations be paid, which will be done in proportion to the obligations for assets $\mu+1$ through $m$.  Mathematically we will define the payment utility function $h_i^\mu: [0,\bar p_i] \times \bbr^{(n-1) \times m}_+ \times [\underline q,\overline q] \to \bbr$ by
\begin{align*}
h_i^\mu(p_i;y_{-i},q) &:= -\frac{1}{2}\left(c_i - p_i\right)^\T \diag\left(\left[\frac{q_1}{c_i^1 - s_i^1} , \dots , \frac{q_{\mu}}{c_i^{\mu} - s_i^{\mu}} ,  \frac{q_{\mu + 1}}{c_i^{\mu + 1} - \pi \bar p_i^{\mu + 1}} , \dots , \frac{q_m}{c_i^m - \pi \bar p_i^m}\right]^\T\right) \left(c_i - p_i\right)\\
c_i &= \bar p_i + \delta \\
s_i^k &= \bar p_i^k \wedge \frac{\left(\sum_{j = 1}^m q_j e_i^j - \sum_{j = 1}^{k-1} q_j s_i^j\right)^+}{q_k} \quad\quad (k = 1,2,...,\mu)\\
\pi &= \frac{\left[\left(\sum_{k = 1}^m q_k \bar p_i^k\right) \wedge \left(\sum_{k = 1}^m q_k e_i^k\right)\right] - \sum_{k = 1}^{\mu} q_k s_i^k}{\sum_{k = 1}^m q_k \bar p_i^k - \sum_{k = 1}^{\mu} q_k s_i^k}\\
e_i^k &= x_i^k + \sum_{j = 1}^n a_{ji}^k \left[\bar p_j^k \wedge y_j^{k}\right] \quad\quad (k = 1,2,...,m).
\end{align*}
As with the surplus payment utility function, we choose the quadratic form to create ellipsoidal level sets with center $c_i$ (above $\bar p_i$), and such that the gradient of $h_i$ is $q$ at the point on the feasible line where all wealth is in asset $1$ if less than $\bar p_i^1$, or $\bar p_i^1$ wealth is in asset $1$, and so on until asset $\mu$, and the remaining assets are along the proportionality line.
\end{example}

\subsection{Sample utility functions}\label{Sec:u}
Here we will consider the details of three possible, meaningful, options for the choice of utility functions $u_i$ in~\eqref{Eq:holding}.  These are the utility function that leads to minimizing the total size of transfers, the utility function that prioritizes holding a specific asset, and the utility function given by the pre-fire sale priced final wealth of the firm.  All three of these sample utility functions satisfy Assumption~\ref{Ass:existence} and the conditions of Algorithm~\ref{Alg:clearing}.

\begin{example}\label{Ex:min-trading}
Consider the case where firms wish to make the smallest possible trades in order to meet their obligations, and trade no more once that occurs.
Such a setting may be appropriate when a firm is concerned about uncertainty in the rightful exchange rates.  With such uncertainty a firm may choose to minimize their own impact and wait for the market response to take shape before responding.
Essentially, this is a ``wait and see'' approach to investing during a crisis.  Firms would then choose to rebalance over time as prices fluctuate after the crisis studied in this work.  Due to the static nature of the model studied in this work, this allows us to capture the immediate aftermath of a crisis but not the long term effects that may be felt.
We can define the utility function for~\eqref{Eq:holding} by
\[u_i(e_i;y_{-i},q) := -\left\|(x_i^k + \sum_{j = 1}^n a_{ji}^k [\bar p_j^k \wedge y_j^{k}])_k - e_i\right\|_2^2.\] 
That is, the holdings for firm $i$ after trading would be the closest feasible point (based on the Euclidean norm) to the initial network model before trading occurs.
In particular, by definition of the norm, $u_i$ is jointly continuous, strictly concave, and supermodular in its first component.
\end{example}

\begin{example}\label{Ex:max-asset}
Consider the case where a firm may wish to maximize their holdings in a specific asset $k^* \in \{1,2,...,m\}$ at the expense of all other assets.  
In a systemic crises, firms may choose to sell higher risk assets or currencies in order to purchase safer ones in a \emph{flight-to-quality}. 
This has occurred in practice during currency crises such as the 1997 Asian financial crisis when firms bought US dollars due to the collapse of local currencies despite the market moving against them in such transactions.
This can be modeled by a firm wishing to maximize the holdings in the safe asset at the expense of all others.
We can define the utility function for~\eqref{Eq:holding} by
\[u_i(e_i;y_{-i},q) := e_i^{k^*} - \bar p_i^{k^*}.\] 
That is, firm $i$ will solely seek to maximize their holdings in asset $k^*$, without consideration of any other assets.  In particular, this is trivially jointly continuous, concave, and supermodular in its first component. 
\end{example}

\begin{example}\label{Ex:wealth}
Consider the case where firms wish to maximize their own net worth (in the num\'{e}raire) given the pre-fire sale prices.  
Such a setting is appropriate when a firm has the belief that the pre-fire sale prices are the ``true'' value of the assets.  In such a view, any change from this price is due to the current crisis, but will rebound to the pre-fire sale prices after the crisis is over.  Thus a firm would wish to purchase assets at a discount (or sell at a premium) in order to obtain a good deal.
In this case we seek to maximize 
\[u_i(e_i;y_{-i},q) := \left(e_i-\bar p_i\right)^\T F(0),\]
which is jointly continuous, concave, and supermodular on the domain of interest.  Additionally, under the condition that $P_i$ is a singleton (as in the assumptions of Lemma~\ref{Lemma:monotonic}), we can recover the resultant utility maximizer $Y_i(y,q)$ is unique so long as $q \neq \lambda F(0)$ for every $\lambda \in \bbr_{++}$. 

Compare this utility function to the welfare maximizing utility, i.e., when a regulator wishes to maximize the welfare of the aggregate system of financial institutions as measured by $\sum_{i = 1}^n u_i(e_i;y_{-i},q)$ over all firm holdings $e_i$ simultaneously.  With this utility function, since firms are purchasing asset $k$ if $q_k < F_k(0)$ and selling if $q_k > F_k(0)$, the Nash equilibrium clearing prices will correspond with the welfare maximizing clearing prices since firm behavior will be identical under both considerations.
\end{example}

\section{Numerical case studies}\label{Sec:casestudy}
In this section we will consider two numerical implementations of the financial contagion framework considered in the prior sections.  The first study is a toy implementation of the example payment utility and utility functions presented in Section~\ref{Sec:examples} to demonstrate how they affect the \emph{attained} equilibrium prices.  We will then implement a brief study of the European financial system, calibrated with data from the European Banking Authority, to compare the equilibrium solution under a single currency (as in \cite{EN01}) to the counterfactual under which the Greek drachma were reinstated during an actualized Grexit event.

\begin{example}\label{Ex:2asset_compare}
For a first illustrative example, consider a network with two currencies.  As assumed throughout this work, let the first currency act as the num\'{e}raire asset, i.e., $F_1 \equiv 1$.  To model the market impacts, let the inverse demand function for the second currency be given by
\begin{align*}
F_2(z) &= \begin{cases} \hat f(z_2) &\text{if } z_2 \geq 0\\ \frac{1}{\hat f\left(\alpha^{-1}(-z_2)\right)} &\text{if } z_2 < 0\end{cases} \quad \text{ for } \quad \hat f(z) = \frac{3\tan^{-1}\left(-z\right) + 2\pi}{2\pi}
\end{align*}
where $\alpha(z_2) := z_2 \hat f(z_2)$ is the number of units of the first currency being purchased when $z_2 \in \bbr$ units of the second currency are being sold.  See Remark~\ref{Rem:idf2asset} for a discussion of the symmetry argument inherent in this choice of inverse demand function.

We will consider a system with 20 firms and a societal node, as introduced in Corollary~\ref{Cor:monotonic-unique}.  As this is an illustrative example only, we will consider a single realization of a random financial network.  In both currencies, independently, each pair of firms has a 25\% probability of having a connection of size 1.  Additionally, every firm owes 1 of both currencies to the external node, which owes nothing back into the system.  All firms begin with i.i.d.\ random endowments uniformly chosen between 0 and 20, which is then split evenly between the two currencies.   For demonstration purposes to show how the different regulation schemes and utility functions alter the equilibrium of the system as the initial shock $q_0$ varies, we consider the surplus, priority, and proportional regulation schemes (Examples~\ref{Ex:surplus} and \ref{Ex:priority-proportional} with $\mu = 2$ and $\mu = 0$ respectively) under the minimum trading utility function (Example~\ref{Ex:min-trading}) and value maximizing utility function (Example~\ref{Ex:wealth}).  In each scenario the societal node will follow the minimum trading utility function.  
Additionally, for illustration of the set of clearing solutions, we will also show all equilibrium prices without any initial shock, i.e.\ $q_0 = F(0) = (1,1)^\T$, for the different regulation schemes and utility functions.

Figure~\ref{Fig:2asset_compare} displays the prices attained from the t\^atonnement process $q^*_2(q_0)$ given an initial price of $q_0 = (1,q_{0,2})^\T$ both with and without market impacts.  We note that the attained process need not be continuous in the initial price $q_0$.  Note that only a single curve for the value maximizing utility is shown as all three payment utility schemes produce virtually indistinguishable curves under that utility.  Further note that under the value maximizing utility the \emph{unique} equilibrium price is given by the unshocked price $F(0) = (1,1)^\T$ for nearly any initial price $q_0$ and mitigates the shock for any initial price.  Additionally, the value maximizing utility produces a continuous equilibrium response as a function of the shocked price $q_0$.  
We would also like to point out that the attained clearing prices need not be continuous (as also demonstrated in Example~\ref{Ex:nonunique}).  It appears that all three regulatory environments jump equilibria at low values of $q_{0,2}$.  Specifically, under the minimal trading utility function, the surplus payment utility function jumps equilibria at $q_{0,2}^s \approx 0.605$ from $q_2^*(q_0^s+\epsilon) \approx 2.695$ to $q_2^*(q_0^s-\epsilon) \approx 0.364$.  Similarly, the priority payment utility function jumps equilibria at $q_{0,2}^{\mu=2} \approx 1.413$ from $q_2^*(q_0^{\mu=2}+\epsilon) \approx 2.956$ to $q_2^*(q_0^{\mu=2}-\epsilon) \approx 0.285$.  Finally, the proportional payment utility function jumps equilibria at $q_{0,2}^{\mu=0} \approx 0.875$ from $q_2^*(q_0^{\mu=0}+\epsilon) \approx 3.001$ to $q_2^*(q_0^{\mu=0}-\epsilon) \approx 0.331$.
\begin{figure}[h]
\centering
\includegraphics[width=0.75\textwidth]{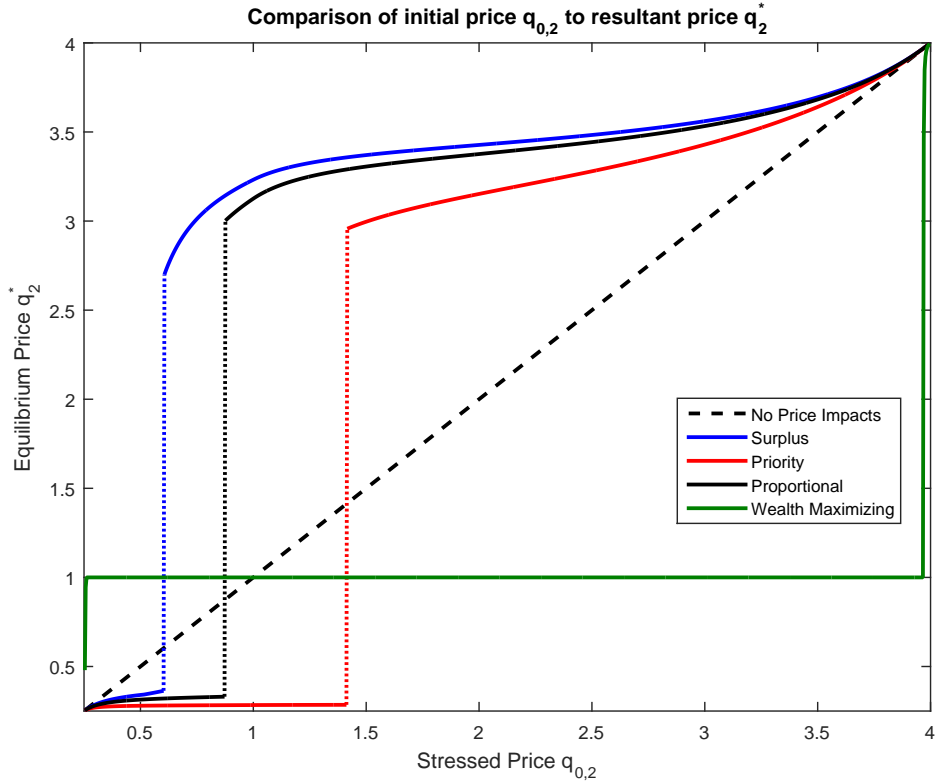}
\caption{Example~\ref{Ex:2asset_compare}: A comparison of the attained clearing prices both without market impacts and with market impacts under surplus, priority, and proportional regulation schemes in a 2 currency system under minimum trading and value maximizing utility.}
\label{Fig:2asset_compare}
\end{figure}

For demonstrative purposes, we display the set of all clearing prices and defaulting banks under the four scenarios in Table~\ref{Table:2asset_compare} under no initial shock, i.e., $q_0 = F(0)$.  We would like to point out that the priority regulatory scheme, under the minimum trading utility function, results in a unique equilibrium in this setting, with the equilibrium price being near the lower bound $\underline q_2 = \frac{1}{4}$.  This is due to the forced liquidation of currency 2 for currency 1, which creates a significant asymmetry in the trading not present in, e.g., the proportional regulation scheme.  We further note that the surplus regulation scheme, though symmetric in construction, is asymmetric in equilibrium.  This demonstrates the concept that the particular realization of the network plays a significant role in directing symmetry (i.e., higher liabilities or lower assets in one currency will skew the equilibrium results).  Finally, the proportional regulation scheme results in multiple clearing solutions thus providing a counterexample to uniqueness of the joint clearing holdings and prices.  Notably, choosing different regulatory and utility functions causes different firms to default in equilibrium.  

\begin{table}
\centering
\begin{tabular}{|l|l|c||*{8}{c|}}
\hline
\textbf{Utility} & \textbf{Regulation} & \textbf{Price} $q_2^*$ & 7 & 8 & 9 & 10 & 13 & 16 & 17 & 19 \\
\hline
\hline
\multirow{5}{*}{Minimum trading} & Surplus & 3.2324 & X & X & X & X & & & & X \\ \cline{2-11}
                                 & Priority & 0.2812 & X & X & & & X & X & X & X \\ \cline{2-11}
                                 & \multirow{3}{*}{Proportional} & 0.3340 & X & X & & & X & X & X & X \\ \cline{3-11}
                                 &                               & 0.7625 & X & X & & & X & & X & X \\ \cline{3-11}
                                 &                               & 3.1252 & X & X & X & X & & & & X \\ \hline
Value maximizing & Proportional & 1.0000 & X & X & & & X & & X & X \\
\hline   
\end{tabular}
\caption{Example~\ref{Ex:2asset_compare}: A comparison of clearing price and defaulting firms without initial shock, i.e., $q_0 = F(0)$, under the surplus, priority, and proportional regulation schemes in a 2 currency system under minimum trading and value maximizing utility.}\label{Table:2asset_compare}
\end{table}
\end{example}

\begin{example}\label{Ex:Greece}
Let us now consider an example calibrated from data.  We will calibrate a network model to the 2011 European banking dataset from EBA that has been used in prior studies (e.g., \cite{GV16,CLY14}) under the financial contagion framework of~\cite{EN01}.  Though we utilize this dataset to have a more realistic network, the approach for calibration still requires heuristics, as such this example is still for demonstrative purposes only.\fn{Due to complications with the calibration methodology, we only consider 87 of the 90 institutions. DE029, LU45, and SI058 were not included in this analysis.}

As a stylized bank balance sheet, we will consider two categories of assets: \emph{interbank assets} $\sum_{j = 1}^n L_{ji}$ and \emph{endowments} $x_i$.  We will additionally consider three categories of liabilities: \emph{interbank liabilities} $\sum_{j = 1}^n L_{ij}$, \emph{external liabilities} $L_{i0}$, and \emph{capital} $c_i$.  First, we will briefly discuss how to calibrate the \cite{EN01} model, i.e., when all values are denominated in the num\'{e}raire asset only.  The EBA dataset provides information on the total assets $T_i$, capital $c_i$, and interbank liabilities $\sum_{j = 1}^n L_{ij}$.  To determine the variables necessary for the \cite{EN01} model we will assume, as in \cite{CLY14,GY14}, that the interbank liabilities equal interbank assets $\sum_{j = 1}^n L_{ij} = \sum_{j = 1}^n L_{ji}$.  
We will, however, modify this condition slightly as discussed in \cite{GV16}; we will perturb the interbank assets a small amount to satisfy a technical condition of \cite{GV16}.
Additionally, we assume that all assets not a part of the interbank assets are endowments and all liabilities not capital or owed to other banks are owed to the societal node $0$.  Under these assumptions, given the provided values, we determine the remainder of the stylized balance sheet via
\begin{align*}
x_i &:= T_i - \sum_{j = 1}^n L_{ij}, \quad L_{i0} := T_i - \sum_{j = 1}^n L_{ij} - c_i, \;\; \text{ and } \;\; \bar p_i := L_{i0} + \sum_{j = 1}^n L_{ij}.
\end{align*}
Under this calibration, the net worth of firm $i$ is equal to its capital, i.e., $c_i = T_i - \bar p_i$.

In order to complete the \cite{EN01} system, we need the full nominal liabilities matrix $L$.  This, however, is not provided in the EBA dataset.  Thus we will utilize the methodology of \cite{GV16} in order to estimate one such matrix consistent with the asset and liability data discussed above.  We consider a single realization of the nominal liabilities matrix $L$ given the algorithm of \cite{GV16} with parameters $p = 0.5$, $\text{thinning} = 10^4$, $n_{\text{burn-in}} = 10^9$, and $\lambda = \frac{p n (n-1)}{\sum_{i = 1}^n \sum_{j = 1}^n L_{ij}} \approx 0.00122$.

First, as a baseline model, we run the financial contagion model of \cite{EN01} to determine the ``factual'' response in the scenario that Greece remains in the Eurozone and thus only a single currency is utilized.  In this scenario, assuming no external stresses, we find that none of the 87 banks would default on its obligations; this comports with reality since none of the firms failed in late 2011.  Additionally, as this model only considers a single asset, there are no fire sales evidenced either.

Now, we wish to consider the counterfactual scenario in which Greece were not a member of the Eurozone and had its own currency, the drachma, once more (i.e., the Grexit scenario).  In order to update the calibration to include both the euro and drachma, we need to consider also 
the total (non-sovereign) exposures that each bank has to Greece $GE_i$.\fn{Sovereign exposures to Greece were orders of magnitude smaller and thus their inclusion would not have significantly affected the final model.}  For notational simplicity let $N = \{1,2,...,87\}$ be the set of all banks and $G \subseteq N$ be the set of the six Greek banks in the EBA dataset.  Then using the calibrated assets and liabilities to the \cite{EN01} framework (henceforth denoted $x^{EN}$ and $L^{EN}$) we update the assets and liabilities to be
\begin{align*}
x_i^1 &:= x_i^{EN} - GE_i, &x_i^2 &:= GE_i \quad\quad &\forall& i \in N \backslash G\\
x_i^1 &:= 0, &x_i^2 &:= x_i^{EN} \quad\quad &\forall& i \in G\\
L_{ij}^1 &:= L_{ij}^{EN}, &L_{ij}^2 &:= 0 \quad\quad &\forall& i \in N \backslash G \; \forall j \in N \cup \{0\}\\
L_{ij}^1 &:= L_{ij}^{EN}, &L_{ij}^2 &:= 0 \quad\quad &\forall& i \in N \; \forall j \in N \backslash G\\
L_{ij}^1 &:= 0, &L_{ij}^2 &:= L_{ij}^{EN} \quad\quad &\forall& i \in G \; \forall j \in G \cup \{0\}.
\end{align*}
where the first asset is the euro and the second is the drachma.  That is, the assets of the non-Greek banks are denominated in drachmas for the amount that was exposed to Greece and the rest remains denominated in euros.  In contrast, all endowments held by Greek banks are re-denominated in the drachma.  Additionally, obligations from a Eurozone bank to the societal node is denominated in the euro and all obligations from a Greek bank to the societal node is denominated in drachmas.  Finally, all interbank liabilities between two Greek banks is re-denominated in drachmas, otherwise all interbank liabilities remain in euros.

In incorporating two assets we need to discuss the inverse demand function.  Consider an inverse demand function of the form of that in Example~\ref{Ex:2asset_compare}.  That is, let the first asset (euro) acts as the num\'{e}raire asset, i.e., $F_1 \equiv 1$, and let the inverse demand function for the second asset (drachma) be given by
\begin{align*}
F_2(z) &= \begin{cases} \hat f(z_2) &\text{if } z_2 \geq 0\\ \frac{1}{\hat f\left(\alpha^{-1}(-z_2)\right)} &\text{if } z_2 < 0\end{cases} \quad \text{ for } \quad \hat f(z) = \frac{4\tan^{-1}\left(-bz\right) + 3\pi}{3\pi}
\end{align*}
where $b \geq 0$ is the market impact parameter and $\alpha(z_2) := z_2 \hat f(z_2)$ is the number of units of euros being purchased when $z_2 \in \bbr$ units of drachmas are being sold.  See Remark~\ref{Rem:idf2asset} for a discussion of the symmetry argument inherent in this choice of inverse demand function.  We will first consider setting the market impacts to a fixed level, i.e.\ $b = 10^{-4}$, then considering the effects of changing the price impacts.

Finally, we need to consider some payment utility and utility functions for this setting.  We will consider that all firms will follow a priority regulation scheme (Example~\ref{Ex:priority-proportional} with $\mu = 2$) in which they prioritize obligations in the local currency due to the preferences of the regulators.  That is, Eurozone banks will prioritize payments in the euro and Greek banks will prioritize payments in the drachma.  Additionally, we will assume that all firms (and the societal node) will follow the minimal trading utility function (Example~\ref{Ex:min-trading}).  This follows from the presupposition that the initial exchange rate (without loss of generality set to $F(0) = (1,1)^\T$) would not be trusted by the various institutions due to fear of fire sales of the new drachma.  Therefore, due to uncertainty in the ``true'' exchange rate, firms will be conservative and do as little trading as necessary. If instead we supposed that firms would want to maximize their assets in the euro as a flight-to-stability (Example~\ref{Ex:max-asset}) then we would see a total collapse of the drachma value but similar final results.

Now, with the multiple currency network calibrated to the setting that Greece was forced out of the Eurozone, we can simulate this systemic event.  For this case study we will only consider the setting without an initial shock, i.e.\ $q_0 = F(0)$ and $\gamma_0 = 0$, and with price impacts given by $b = 10^{-4}$.  Due to the choice of priority regulation scheme, the response to external stresses to the value of the drachma (i.e., with $q_{0,2} < 1$) would only marginally impact the final equilibrium.
Figure~\ref{Fig:2asset_Greece} displays the updated prices $F(\sum_{i = 1}^n (x_i + [\bar p_i \wedge y_i^*(q)] - y_i^*(q)))$ given an initial price of $q = (1,q_2)^\T$ after one iteration of the fixed point problem.  We note that the resultant curve is continuous because of the  continuity of the both the inverse demand function and the unique holdings $y^*(q)$ (see the proof of Corollary~\ref{Cor:existence}).  The marked point on the curve shows the equilibrium price of $q^* = (1,0.44331)$, i.e., after clearing the value of the drachma would fall to 44.331\% of its former value compared to the euro.  At this equilibrium price, we found that three banks would fail, though none of these banks were situated in Greece.  This is due to a large amount of Greek bank liabilities (interbank and to society) being drachma denominated and thus insulated from the fall in drachma value.\fn{Though the greater Greek economy would likely suffer under such a large currency move.}  However we found that both banks from Cyprus included in this dataset (Marfin Popular Bank and the Bank of Cyprus) and the Banco Comercial Portugues would fail due to large exposures to Greece.\fn{We would like to note that all three of these defaulting banks received a bailout or government intervention in either 2012 or 2013.}  While the Cypriot banks had, in relative terms, an order of magnitude more exposures to Greece than any other non-Greek bank, the Banco Comercial Portugues had the third highest relative exposure to Greece.
Thus we are able to see that a crisis focused on Greece, with an endogenous stress to the rest of the Eurozone, is able to spread to institutions in the rest of the Eurozone.  In particular, if our study had coupled the Grexit event with some exogenous stress to the different banks balance sheets, as would likely occur in such an event, we would find a larger number of defaults in both the Eurozone and Greece.

We wish to finish this example by considering the effects of changing the price impact parameter $b$ which previously was fixed at $10^{-4}$.  All other parameters are kept constant from the previous considerations.  Notably this includes the assumption that there is no initial crisis and all price movements are the result of the actions of the firms under consideration.  Figure~\ref{Fig:2asset_Greece_impact} displays the attained equilibrium prices under changes to the price impact parameter $b$. Notably even a small level of price impacts causes a large drop in the value of the Greek drachma in relation to the euro.  This provides us with a level of confidence in the determined results we found for the price impact $b = 10^{-4}$ even though this inverse demand function was not calibrated to data in the manner that the balance sheets were. 

\begin{figure}[h]
\centering
\begin{subfigure}[t]{.47\textwidth}
\vskip 0pt
\includegraphics[width=\textwidth]{2asset_Greece.eps}
\caption{A graphical representation of the output of the fire sale prices (with price impact $b = 10^{-4}$).  The marked point indicates the unique equilibrium price.  The text box provides key information on the equilibrium price and banking failures experienced.}
\label{Fig:2asset_Greece}
\end{subfigure}
~
\begin{subfigure}[t]{.47\textwidth}
\vskip 0pt
\includegraphics[width=\textwidth]{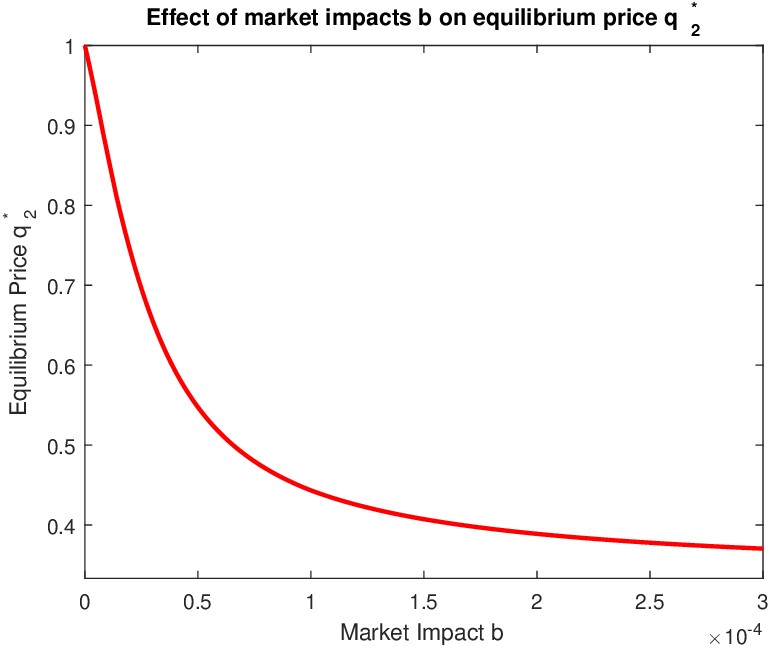}
\caption{The effects of the price impact parameter $b \geq 0$ on the attained clearing prices of the Greek drachma denominated in euros.}
\label{Fig:2asset_Greece_impact}
\end{subfigure}
\caption{Example~\ref{Ex:Greece}: Graphical representations for fire sale prices of the Greek drachma denominated in euros under a priority regulation scheme and minimum trading utility without any initial stress.}
\end{figure}

\end{example}

\section{Conclusion}\label{Sec:conclusion}

In this paper we considered an extension of the financial contagion model of \ci{EN01} to allow for obligations in more than one asset.  In doing so, we have written a mathematical model that incorporates more realistic elements to financial contagion including obligations in multiple currencies and allowing for solvent firms to purchase or sell assets beyond those required to satisfy obligations (as is generally assumed in, e.g., \cite{CFS05,AFM16,feinstein2015illiquid,feinstein2016leverage}).  
Under markets without price impacts, we proved the existence and uniqueness of the equilibrium portfolio holdings in which each firm is a utility maximizer.  We then generalized this result to prove existence of clearing prices in markets with price impacts.  Additionally, we consider the t\^atonnement process to determine which equilibria to which the market would converge in the $2$ asset setting.
Numerical case studies were undertaken to demonstrate the utility of the proposed model and how the choice of payment regulatory framework may impact, e.g., the realized exchange rates.  In particular, we consider a stylized example of the European, and specifically Greek, financial system under the counterfactual condition that the Greek drachma were reinstated.

\appendix
\section{Proofs for Section~\ref{Sec:model}}
\begin{proof}[Proof of Lemma~\ref{Lemma:monotonic}]
Fix $q \in \bbr^m_{++}$.
\begin{enumerate}
\item Define $G_i: [0,\bar p_i] \to \bbr$ to be the linear mapping $G_i(p_i) = q^\T p_i$ for any $p_i \in [0,\bar p_i]$.  Noting that the (convex) budget constraint of $P_i(y,q)$ is equivalently given by 
\[p_i \in G_i^{-1}\left(\left(-\infty,\sum_{k = 1}^m q_k \left(x_i^k + \sum_{j = 1}^n a_{ji}^k \left[\bar p_j^k \wedge y_j^k\right]\right)\right]\right).\]
Further, the upper bound from the budget constraint is nondecreasing in $y$.  Utilizing the strict concavity of the payment utility function $h_i$ to guarantee the uniqueness of the maximizer $P_i(y,q)$ for any $y \in \bbr^{n \times m}_+$, it follows that $P_i(\cdot,q)$ is nondecreasing by Corollary 2(ii) of \cite{Quah2007}.

Now we wish to show that $Y_i(\cdot,q)$ is nondecreasing as well.  That is, $Y_i(y,q) \leq Y_i(y',q)$ for any $y,y' \in \bbr^{n \times m}_+$ with $y \leq y'$.  Take $y,y' \in \bbr^{n \times m}_+$ with $y \leq y'$.  If $P_i(y,q) \neq \bar p_i$ then, by construction and the monotonicity of the payment function $P_i(\cdot,q)$, we find
\[Y_i(y,q) = P_i(y,q) \leq P_i(y',q) \leq Y_i(y',q).\]
Let $P_i(y,q) = \bar p_i$ (and thus $P_i(y',q) = \bar p_i$ as well by the monotonicity of $P_i(\cdot,q)$).  Now define $G_i: \left(\bar p_i + \bbr^m_+\right) \to \bbr$ as the same linear map as above, i.e., $G_i(e_i) = q^\T e_i$.  As with the payment function, the feasible regions for $Y_i(y,q)$ and $Y_i(y',q)$ can, respectively, be provided by
\begin{align*}
G_i^{-1}&\left(\left(-\infty,\sum_{k = 1}^m q_k \left(x_i^k + \sum_{j = 1}^n a_{ji}^k \left[\bar p_j^k \wedge y_j^k\right]\right)\right]\right),\\
G_i^{-1}&\left(\left(-\infty,\sum_{k = 1}^m q_k \left(x_i^k + \sum_{j = 1}^n a_{ji}^k \left[\bar p_j^k \wedge y_j'^k\right]\right)\right]\right).
\end{align*}
Again, the upper bound for this interval is nondecreasing in the portfolio holdings parameter.  Under the assumption that $Y_i(y,q)$ and $Y_i(y',q)$ are unique maximizers, we apply Corollary 2(ii) of \cite{Quah2007} to find $Y_i(y,q) \leq Y_i(y',q)$.

Finally, we apply the Tarski fixed point theorem (see, e.g., Theorem 11.E of \cite{Z86}) to the mapping $Y(\cdot,q): \bbr^{n \times m}_+ \to \bbr^{n \times m}_+$ to recover the result.
\item From~\eqref{Lemma:monotonic-exist} we know that $(y_i^{\uparrow}(q)-\bar p_i)^+ \geq (y_i^{\downarrow}(q)-\bar p_i)^+$.  By way of contradiction, assume there exists an institution $i$ and asset $k$ such that $(y_i^{\uparrow k}(q)-\bar p_i^k)^+ > (y_i^{\downarrow k}(q)-\bar p_i^k)^+$.  This immediately implies $\sum_{i = 1}^n q^{\T}(y_i^{\uparrow}(q)-\bar p_i)^+ > \sum_{i = 1}^n q^{\T}(y_i^{\downarrow}(q)-\bar p_i)^+$ by $q_k > 0$ for every asset $k$.  However,
\begin{align*}
\sum_{i = 1}^n q^{\T}(y_i^{\uparrow}(q) - \bar p_i)^+ &= \sum_{i = 1}^n q^{\T}(y_i^{\uparrow}(q) - [\bar p_i \wedge y_i^{\uparrow}(q)])\\
&= \sum_{i = 1}^n (q^{\T}x_i + \sum_{k = 1}^m q_k \sum_{j = 1}^n a_{ji}^k [\bar p_j^k \wedge y_j^{\uparrow k}(q)] - q^{\T}[\bar p_i \wedge y_i^{\uparrow}(q)])\\
&= \sum_{i = 1}^n q^{\T}x_i + \sum_{k = 1}^n q_k \sum_{j = 1}^n [\bar p_j^k \wedge y_j^{\uparrow k}(q)] \sum_{i = 1}^n a_{ji}^k - \sum_{i = 1}^n q^{\T}[\bar p_i \wedge y_i^{\uparrow}(q)]\\
&= \sum_{i = 1}^n q^{\T}x_i + \sum_{j = 1}^n q^{\T}[\bar p_j \wedge y_j^{\uparrow}(q)] - \sum_{i = 1}^n q^{\T}[\bar p_i \wedge y_i^{\uparrow}(q)]\\
&= \sum_{i = 1}^n q^{\T}x_i = \sum_{i = 1}^n q^{\T}(y_i^{\downarrow}(q) - \bar p_i)^+ 
\end{align*}
where the last equality follows by applying the same operations to $y_i^{\downarrow}$ in the reverse order.
This provides our contradiction and thus $(y_i^{\uparrow}(q)-\bar p_i)^+ = (y_i^{\downarrow}(q)-\bar p_i)^+$ for every bank $i$.
\end{enumerate}
\end{proof}

\begin{proof}[Proof of Corollary~\ref{Cor:monotonic-unique}]
First, if bank $i$ has positive equity, then by Lemma~\ref{Lemma:monotonic}\eqref{Lemma:monotonic-equity} it immediately follows that $y_i^{\uparrow}(q) = y_i^{\downarrow}(q)$.  
In particular this must be true for node $0$ as it has positive equity by definition.  
For notational purposes, let $E^+ := \{i \in \{1,2,...,n\} \; | \; y_i^{\downarrow}(q) \geq \bar p_i\}$ be the set of firms with positive equity.  
Let us assume there exists some firm $i \not\in E^+$ and asset $k$ such that $y_i^{\uparrow k}(q) > y_i^{\downarrow k}(q)$, then immediately the mark-to-market value of the equity of the societal node $0$ satisfies
\begin{align*}
q^{\T}y_0^{\uparrow}(q) &= \sum_{k = 1}^m q_k \sum_{j \in E^+} a_{j0}^k \bar p_j^k + \sum_{k = 1}^m q_k \sum_{j \not\in E^+} a_{j0}^k y_j^{\uparrow k}(q)\\
&> \sum_{k = 1}^m q_k \sum_{j \in E^+} a_{j0}^k \bar p_j^k + \sum_{k = 1}^m q_k \sum_{j \not\in E^+} a_{j0}^k y_j^{\downarrow k}(q) = q^{\T}y_0^{\downarrow}(q).
\end{align*}
But this is a contradiction to $y_0^{\uparrow}(q) = y_0^{\downarrow}(q)$. 
\end{proof}

\begin{proof}[Proof of Corollary~\ref{Cor:existence}]
First, we wish to prove that, given uniqueness (guaranteed by Corollary~\ref{Cor:monotonic-unique}) of the portfolio holdings under a fixed price $q$, the equilibrium holdings $y^*: [\underline q,\overline q] \to \bbr^{n \times m}_+$ are continuous.  Theorem A.2 of~\cite{feinstein2014measures} guarantees that the graph of $y^*$ is closed in the product topology.  Now we note that the range space that the holdings can attain is, in fact, the convex and compact set $\prod_{i = 1}^n \bar E_i$ where
\[\bar E_i = \prod_{k = 1}^m \left[0,\frac{1}{\underline q_k} \sum_{l = 1}^m \overline q_l \left(x_i^l + \sum_{j = 1}^n a_{ji}^l \bar p_j^l\right)\right].\]
Therefore, by the closed graph theorem (see, e.g.,~\cite[Theorem 2.58]{AB07}) continuity is proven.  This allows us to directly apply the Brouwer fixed point theorem (see, e.g.,~\cite[Corollary 17.56]{AB07}) to find an equilibrium price
\[q^* = F\left(\sum_{i = 1}^n (x_i + [\bar p_i \wedge y_i^*(q^*)] - y_i^*(q^*))\right).\]
\end{proof}

\begin{proof}[Proof of Proposition~\ref{Prop:tatonnement}]
Define $\alpha: [\underline q,\overline q] \to \bbr^m$ by $\alpha(q) := \sum_{i = 1}^n \left(x_i + [\bar p_i \wedge y_i(q)] - y_i(q)\right)$.
Additionally, consider $V: [\underline q,\overline q] \to \bbr$ to be provided by
\[V(q) := \gamma_0^\T q + \Alpha(q) - \sum_{k = 1}^m \int_1^{q_k} f_k^{-1}(p) dp\]
where $\Alpha: [\underline q,\overline q] \to \bbr$ is defined as the multivariate function with gradient $\alpha$.  For the moment we will assume that $\Alpha$ exists, at the end of this proof we show this is the case under the restricted $m = 2$ asset setting.
With this construction we find that $\frac{d}{dt}V(q_t)$ is negative semidefinite for any trajectory $q_t$ of the t\^atonnement process, i.e.,
\begin{align*}
\frac{d}{dt}V(q_t) &= \left(\gamma_0 + \alpha(q_t) - F^{-1}(q_t)\right)^\T \left(F(\gamma_0 + \alpha(q_t)) - q_t\right) \leq 0.
\end{align*}
This is because $[\gamma_0 + \alpha(q) - F^{-1}(q)]_k \geq 0$ if and only if $F_k(\gamma_0 + \alpha(q)) \leq q_k$ since $F_k(\gamma_0 + \alpha(q)) = F_k(\gamma_0 + \alpha(q) - F^{-1}(q) + F^{-1}(q))$ and the monotonicity of the inverse demand function.  In fact, $\frac{d}{dt}V(q_t) = 0$ if and only if $q_t$ is an equilibrium price due to the same preceding argument.  By LaSalle's invariance principle (see, e.g., \cite[Theorem 4.4]{khalil2002}), the set of accumulation points of any trajectory is equivalent to the set of equilibrium prices.  Further, since $q_t \in [\underline q,\overline q]$ for every time $t \geq 0$, the t\^atonnement process approaches the set of clearing prices as $t \to \infty$.

Finally, we wish to guarantee the existence of the function $\Alpha: [\underline q,\overline q] \to \bbr$ so that its gradient is equal to $\alpha$.  For this purpose we restrict ourselves to the $m = 2$ asset setting since, functionally, we can consider our input $q$ to be its second argument only (due to our choice of num\'{e}raire).  That is, we can consider instead $V(q_2) = \gamma_{0,2}q_2 + \int_1^{q_2} \alpha_2(p) dp - \int_1^{q_2} f_2^{-1}(p) dp$.
\end{proof}


\bibliographystyle{plainnat}
\bibliography{bibtex2}

\end{document}